\documentclass{article}

\usepackage[final,nonatbib]{neurips_2021}

\usepackage[utf8]{inputenc} %
\usepackage[T1]{fontenc}    %
\usepackage{url}            %
\usepackage{booktabs}       %
\usepackage{amsfonts}       %
\usepackage{nicefrac}       %
\usepackage{microtype}      %
\usepackage{microtype}
\usepackage{crossreftools}

\usepackage{mathscinet}
\pdfoutput=1
\usepackage{amsmath,amsthm,wrapfig}
	\usepackage[utf8]{inputenc} 
\usepackage{amsmath, amssymb,bm, cases, mathtools, thmtools}
\usepackage{verbatim}
\usepackage{graphicx}\graphicspath{{figures/}}
\usepackage{multicol}
\usepackage{tabularx}
\usepackage[usenames,dvipsnames]{xcolor}
\usepackage{mathrsfs} 
\usepackage{caption}
\usepackage{algorithm}
\usepackage{algorithmicx}
\usepackage[noend]{algpseudocode}
\usepackage{array,booktabs,arydshln,xcolor}

\usepackage[%
		minnames=1,maxnames=99,maxcitenames=2,
		style=numeric-comp,%
		sorting=none,
		sortcites=false, %
		doi=false,url=false,
		uniquename=init,
		giveninits=true,
		hyperref,natbib,
		backend=biber]{biblatex}
\renewbibmacro{in:}{%
	\ifentrytype{article}{}{\printtext{\bibstring{in}\intitlepunct}}}

\usepackage[T1]{fontenc}
\usepackage{inconsolata}
\usepackage{dsfont}

\usepackage{listings} %

\usepackage{enumitem}
\usepackage{booktabs}       %
\usepackage{nicefrac}       %

\usepackage{lineno}
\usepackage{bbm}

\usepackage{caption}
\usepackage{subcaption}

\usepackage{datetime}

\DeclareMathAlphabet\EuRoman{U}{eur}{m}{n}
\SetMathAlphabet\EuRoman{bold}{U}{eur}{b}{n}

\declaretheorem[style=plain,numberwithin=section,name=Theorem]{theorem}
\declaretheorem[style=plain,sibling=theorem,name=Lemma]{lemma}

\declaretheorem[style=plain,sibling=theorem,name=Corollary]{corollary}

\declaretheorem[style=definition,sibling=theorem,name=Definition]{definition}

\declaretheorem[style=remark,qed=$\triangleleft$,sibling=theorem,name=Remark]{remark}
\numberwithin{theorem}{section}

\usepackage{xparse}
\usepackage{xstring}
\usepackage{xspace}

\def\[#1\]{\begin{align}#1\end{align}}
\def\*[#1\]{\begin{align*}#1\end{align*}}

\newcommand{\minf}[1]{I(#1)}

\newcommand{\entr}[1]{\mathrm{H}(#1)}

\newcommand{\centr}[2]{\mathrm{H}^{#1}(#2)}

\newcommand{\cPr}[2]{\Pr^{#1}[#2]}

\newcommand{\Dist}{\mathcal D}
\newcommand{\dataspace}{\mathcal Z}
\newcommand\optparen[1]{\ifthenelse{\equal{#1}{}}{}{(#1)}}
\newcommand{\RiskChar}{R}
\newcommand{\Risk}[2]{\RiskChar_{#1}\optparen{#2}}
\newcommand{\EmpRisk}[2]{\hat \RiskChar_{#1}\optparen{#2}}

\newcommand{\dist}{\ \sim\ }

\newcommand{\Naturals}{\mathbb{N}}

\newcommand{\Reals}{\mathbb{R}}

\newcommand{\Nats}{\mathbb{N}}

\newcommand{\as}{\textrm{a.s.}}

\newcommand{\downto}{\!\downarrow\!}

\newcommand{\dee}{\mathrm{d}}

\newcommand{\inspace}{\mathcal X}
\newcommand{\outspace}{\mathcal Y}

\DeclareMathOperator*{\newlim}{\mathrm{lim}\vphantom{\mathrm{infsup}}}
\DeclareMathOperator*{\newmin}{\mathrm{min}\vphantom{\mathrm{infsup}}}
\DeclareMathOperator*{\newmax}{\mathrm{max}\vphantom{\mathrm{infsup}}}

\DeclareMathOperator*{\newsup}{\mathrm{sup}\vphantom{\mathrm{infsup}}}
\renewcommand{\lim}{\newlim}
\renewcommand{\min}{\newmin}
\renewcommand{\max}{\newmax}

\renewcommand{\sup}{\newsup}

\newcommand{\ProbMeasures}[1]{\mathcal{M}_1(#1)}

\renewcommand{\Pr}{\mathbb{P}}
\def\EE{\mathbb{E}}

\newcommand{\defn}[1]{\textit{#1}}

\newcommand{\cF}{\mathcal F}
\newcommand{\cG}{\mathcal G}

\newcommand{\iid}{i.i.d.}

\newcommand{\KLname}{\mathrm{KL}}

\newcommand{\KL}[2]{\KLname(#1 \,\|\,#2)}

\newcommand{\XX}{\Reals^{\idim}}
\newcommand{\YY}{K}

\newcommand{\loss}{\ell}
\newcommand{\hypothesisclass}{\mathcal{H}}

\newcommand{\e}{\mathrm{e}}

\newcommand{\Alg}{\mathcal{A}}

\newcommand{\unif}[1]{\text{Unif}(#1)}

\newcommand{\bernoulli}{\text{Ber}}

\newcommand{\lcrx}[4][{-1}]{
	\IfEq{#1}{-1}{\left #2 {{{{#3}}}} \right #4}{
   	\IfEq{#1}{0}{#2 {{{{#3}}}} #4}{
	\IfEq{#1}{1}{\bigl #2 {{{{#3}}}} \bigr #4}{
	\IfEq{#1}{2}{\Bigl #2 {{{{#3}}}} \Bigr #4}{
	\IfEq{#1}{3}{\biggl #2 {{{{#3}}}} \biggr #4}{
	\IfEq{#1}{4}{\Biggl #2 {{{{#3}}}} \Biggr #4}{
    \GenericWarning{"4th argument to lcrx must be -1, 0, 1, 2, 3, or 4"}
    }}}}}}}

\newcommand{\sbra}[2][{-1}]{\lcrx[#1] [ {#2} ] }

\newcommand{\rnderiv}[2]{\frac{\text{d} #1}{\text{d} #2}}

\newcommand{\cEE}[1]{\EE^{#1}}

\newcommand{\dminf}[2]{I^{#1} (#2)}

\newcommand{\indic}[1]{ \mathds{1}[#1]}

\newcommand{\supersam}{Z}

\newcommand{\range}[1]{ [#1] }

\newcommand{\samplecomplex}[2]{ \mathcal{M}^{\HS}_\text{prop}(#1,#2) }

\newcommand{\vcd}{d} %

\newcommand{\starnum}{\mathfrak{s}}

\newcommand{\teachdim}{\mathrm{ETD}}

\newcommand{\targetfun}{h^\star}

\newcommand{\binaryentr}[1]{\mathrm{H}_{b}(#1)}

\newcommand{\Ur}[2]{U_{#1 \to #2}}
\newcommand{\Sr}[2]{S_{#1 \to #2}}

\usepackage{xspace}

\title{Towards a Unified Information-Theoretic\\Framework for Generalization}

\renewcommand{\defn}[1]{\emph{#1}}

\newcommand{\EGE}{\ensuremath{\mathrm{EGE}_{\Dist}(\Alg)}}
\newcommand{\IWS}{\ensuremath{\mathrm{IOMI}_{\Dist}(\Alg)}\xspace}
\newcommand{\SZB}{\ensuremath{\mathrm{CMI}_{\Dist}(\Alg)}\xspace}
\newcommand{\eSZB}{\ensuremath{\mathrm{eCMI}_{\Dist}(\loss(\Alg))}\xspace}

\newcommand{\versionspace}[1]{\mathrm{V}_{\hypothesisclass}[#1]}
\usepackage{accents}

\newcommand{\supersamflat}{\tilde{Z}}

\newcommand{\oneinclusiongraph}{\mathcal{G}_{\HS}}

\usepackage{contour}
\usepackage[normalem]{ulem}

\contourlength{0.8pt}

\renewcommand{\underline}[1]{%
  \uline{\phantom{#1}}%
  \llap{\contour{white}{#1}}%
}

\usepackage[hide=true,setmargin=true,marginparwidth=1.2in]{marginalia}

\definecolor{mylinkcolor}{rgb}{0,0,.25}
\usepackage[colorlinks,citecolor=mylinkcolor,urlcolor=mylinkcolor,linkcolor=mylinkcolor]{hyperref}

\definecolor{mylinenocolor}{rgb}{0.65,0.65,0.65}

\definecolor{myeqncolor}{rgb}{0.425,0.425,0.425}
\makeatletter
\let\reftagform@=\tagform@
\def\tagform@#1{\maketag@@@{\ignorespaces\textcolor{myeqncolor}{(#1)}\unskip\@@italiccorr}}
\renewcommand{\eqref}[1]{\textup{\reftagform@{\ref{#1}}}}
\makeatother

\usepackage[capitalize]{cleveref}

\crefname{lemma}{Lemma}{Lemmas}
\crefname{corollary}{Corollary}{Corollaries}
\crefname{theorem}{Theorem}{Theorems}

\crefname{problem}{Conjecture}{Conjectures}

\newtheorem{statement}{Statement}
\crefname{statement}{Statement}{Statement}

\newcommand{\HS}{\mathcal H}

\renewcommand{\XX}{\mathcal X}
\renewcommand{\YY}{\mathcal Y}
\renewcommand{\SS}{S_n}
\renewcommand{\Alg}{\mathcal A_n}
\newcommand{\TheAlg}{\mathcal A}

\author{
	Mahdi Haghifam\thanks{Part of the work was done the author was an intern at Element AI, a ServiceNow company.}\\
	University of Toronto, \\
	Vector Institute \\
	\And
	Gintare Karolina Dziugaite\thanks{This work was carried out while the author was at ServiceNow. It was finalized at Google Brain.}\\
	Mila
	\And
	Shay Moran\\
	Technion,\\
	 Google Research \\
	\And
	Daniel M. Roy \\
	University of Toronto,\\
	Vector Institute\\
}

\begin{document}

\maketitle

\begin{abstract}
In this work, we investigate the expressiveness of the ``conditional mutual information'' (CMI) 
    framework of \citet{steinke2020reasoning} and the prospect of using it to provide 
    a unified framework for proving generalization bounds in the realizable setting.
    We first demonstrate that one can use this framework to express non-trivial (but sub-optimal) bounds 
    for any learning algorithm that outputs hypotheses from a class of bounded VC dimension.
    We then explore two directions of strengthening this bound:
    (i)~Can the CMI framework express \underline{optimal} bounds for VC classes?
    (ii)~Can the CMI framework be used to analyze algorithms whose output hypothesis space is \underline{unrestricted} 
    (i.e.\ has an unbounded VC dimension)?
    
With respect to Item (i) we prove that the CMI framework yields the optimal bound on the expected risk  
    of \emph{Support Vector Machines} (SVMs) for learning halfspaces. 
    This result is an application of our general result showing that \emph{stable} compression schemes \citep{bousquet2020proper} of size $k$ have uniformly bounded CMI of order $O(k)$. 
    
We further show that an inherent limitation of proper learning of VC classes contradicts the existence of a proper learner with constant CMI, and it implies a negative resolution to an open problem of \citet{steinke2020openproblem}.  We further study the CMI of empirical risk minimizers (ERMs) of class~$\HS$ and show that it is possible to output all  consistent classifiers (version space) with bounded CMI \emph{if and only if} $\HS$ has a bounded star number \citep{hanneke2015minimax}. 

With respect to Item (ii) we prove a general reduction showing that ``leave-one-out'' analysis is expressible via the CMI framework. 
    As a corollary we investigate the CMI of the one-inclusion-graph algorithm proposed by \citet{haussler1994predicting}.
    More generally, we  show that the CMI framework is universal in the sense that for \emph{every} consistent algorithm and data distribution, the expected risk vanishes as the number of  samples diverges \emph{if and only if} its evaluated CMI has sublinear growth with the number of samples.
\end{abstract}

\section{Introduction}

In this work, we study the expressiveness of generalization bounds in terms of
information-theoretic measures of dependence between  
the output of a learning algorithm
and input data. Information-theoretic techniques for proving generalization bounds are powerful; they can provide generalization bounds that are algorithm-dependent, data-dependent, and distribution-dependent.  This approach was initiated by \citet{RussoZou15,RussoZou16} and \citet{XuRaginsky2017} and has since been
extended by a number of authors
\citep{raginsky2016information,asadi2018chaining,PensiaJogLoh2018,bassily2018learners,nachum2018direct,nachum2019average,IbrahimEspositoGastpar19,negrea2019information,BuZouVeeravalli2019,roishay,neu2021information}. 
\NA{More recently, attention has shifted to whether these techniques can also characterize worst-case (minimax) rates for certain learning problems.}

Let $\Dist$ be an unknown distribution on a space $\dataspace$,
and let $\HS$ be a set of classifiers. Consider a (randomized) learning algorithm  $\TheAlg = (\Alg)_{n\ge 1}$ that selects an element $\hat{h}$ in $\HS$,
based on $n$ \iid\ samples
$
\smash{S_n \sim \Dist^{\otimes n}}
$, i.e., $\hat{h}=\hat{h}_n=\Alg(S_n)$.
The initial focus of this line of work was on the mutual information $\minf{\Alg(S_n);S_n}$ between the input and the output of a learning problem. This quantity is sometimes referred to as the \emph{input--output mutual information (IOMI)} of an algorithm and denoted by $\IWS$. 
A natural question is 
whether the IOMI framework can provide a sharp characterization of the learnability of Vapnik--Chervonenkis (VC) classes, for which we have strong generalization guarantees. A negative resolution was provided by \citet{bassily2018learners} for the concept class of thresholds in one dimension. Follow up work by \citet{nachum2018direct} extended the argument in \citep{bassily2018learners}, proving the following result:
\begin{theorem}[Thm.~1, \citealt{nachum2018direct}]
\label{thm:limitation-imoi}
For every $\vcd\in \Naturals$ and every $n\geq 2 \vcd^2 $, there exists a finite input space $\inspace$ and a concept class $\HS \subseteq \{0,1\}^\inspace$ of VC-dimension $\vcd$ such that, for all proper and consistent learning algorithm $\Alg$, there exists a realizable distribution $\Dist$ such that $\IWS=\Omega(\vcd \log\log(|\inspace|/\vcd))$. 
\end{theorem}
 \citet{roishay} extended this result even further, showing that, for the class of one-dimensional thresholds over $\{1,\dots,m\}$, $m \in \Naturals$,\footnote{
This concept class can be defined as follows. Let $\mathcal{X}=\{1,...,m\}$. 	Let $k \in \Naturals$ and $h_k: \mathcal{X} \to \{0,1\}$ define as $h_k(x)=\mathds{1}[x> k]$. Then, the class of one-dimensional thresholds over $\{1,...,m\}$ is $\mathcal{H}_m=\{h_k\vert k\in \mathbb{N}\}$.
 } for \emph{every} learning algorithm $\TheAlg$ there exists a realizable distribution such that either the risk (population loss) is large or the $\IWS$ scales with the cardinality of the space, $m$. These results highlight an important limitation of the IOMI framework: given an unbounded input space, for any ``good'' learning algorithm there are always scenarios in which $\IWS$ is unbounded. 
 Therefore, the distribution-free learnability of VC classes cannot be expressed using the IOMI framework.

In this paper, we focus on the ``\emph{conditional mutual information}'' (CMI) framework, 
proposed by \citet{steinke2020reasoning}. 
In order to reason about the generalization error of a learning algorithm, they introduce a super sample that contains the training sample as a random subset and compute the mutual information between the input and output conditional on the super sample (formal definitions are provided in \cref{subsec:cmi-defs}). Improvements of this framework and its application in studying the generalization of specific learning algorithms have been studied in   \citep{CCMI20,haghifam2020sharpened,hellstrom2020generalization,rodriguez2020random,zhou2020individually, hellstrom2020nonvacuous}.

The current paper revolves around the following fundamental question: For which learning problems and learning algorithms is the CMI framework expressive enough to accurately estimate the generalization error? 
We will focus in particular on whether we can recover optimal worst case (minimax) rates for VC classes satisfying certain properties.  The answer to these question provide evidence that the CMI framework provides a unifying framework for studying generalization.

For VC classes, \citet{steinke2020reasoning} 
revealed a stark separation between the CMI framework and IOMI framework. 
They showed the existence of an empirical risk minimization (ERM) algorithm whose CMI is no larger than $\vcd \log n + 2$ for learning every VC class of dimension~$\vcd$ given $n$ i.i.d.\ training samples. 
In contrast to \cref{thm:limitation-imoi}, CMI does not scale with the cardinality of the space. 
However, the bound on the CMI combined with \citeauthor{steinke2020reasoning}'s CMI-based generalization bound, leads to a bound on the expected excess risk that is suboptimal in some cases by a $\log n$ factor.
(For an overview of the known bounds for learning VC classes, please refer to \cref{apx:known-bounds}.)
The suboptimality of their bound prompted \citet{steinke2020openproblem} to conjecture that the CMI bound for proper learners of VC classes can be improved to $O(\vcd)$. Moreover, \citeauthor{steinke2020reasoning} connected CMI framework to the sample compression framework of \citep{littlestone1986relating} by showing that a sample compression $\TheAlg_n$ of size $k$ has $\SZB \leq k \log2n$. Their bound for sample compression schemes is also suboptimal in some cases by a $\log n$ factor.

\subsection{Contributions}
In this paper we extend the reach of the CMI framework by demonstrating its unifying nature for obtaining optimal or near-optimal bounds for the expected excess risk of the various algorithms in the realizable setting.
\begin{enumerate}[leftmargin=1.5em]
\item We demonstrate that one can use the CMI framework to express non-trivial (but sub-optimal) bounds for every improper learning algorithm that outputs a hypothesis from a class with a bounded VC dimension. This is achieved by an empirical variant of CMI  defined by~\citep{steinke2020reasoning}.
\item We study the CMI of SVMs for learning half spaces and show that the CMI framework yields optimal bounds on the expected excess risk. Our bound on the CMI of SVM is an application of our general result giving optimal  CMI bounds for stable sample compression schemes  \citep{bousquet2020proper,hanneke2021stable},
which improve on CMI bounds for general sample compression schemes \citep{steinke2020reasoning} by a $\log n$ factor. 
\item  In the context of proper learning of VC classes, we exhibit VC classes
for which the CMI of any proper learner cannot be bounded by any real-valued function of the VC dimension. Then, we consider VC classes with finite star number \citep{hanneke2015minimax}, and prove the existence of a learner with bounded CMI. Finally, we show that the release of the set of all consistent classifiers in $\HS$ has bounded CMI \emph{if and only if} $\HS$ has finite star number.

\item  We show that CMI framwork is \emph{universal} in the realizable            setting. More precisely, for every data distribution      and consistent learner, the bound on  excess risk obtained by the CMI framework vanishes  if and only if the excess  risk  also vanishes  as the     number training samples  diverges.  We then show
     that any learning algorithm with a ``leave-one-out'' bound of order $O(1/n)$
     yields an evaluated-CMI bound of order $O(\log n)$.
    As an application, we study the classical one-inclusion graph algorithm of \citet{haussler1994predicting} 
    for improper learning of VC classes, and provide a nearly optimal bound on its expected excess risk using the CMI framework. We also prove there exists a randomized one-inclusion graph which learns point functions (singleton) with bounded CMI.

\end{enumerate}

Our results indicate that CMI is a very expressive generalization framework, and one that can tie together existing frameworks. Although most of our results are stated for binary classification in the distribution-free setting, it is interesting to note that the CMI framework is known to provide numerically non-vacuous generalization error guarantees for some modern deep learning models and datasets in the distribution-dependent setting \citep{haghifam2020sharpened,hellstrom2020nonvacuous}. These developments in a range of different problem settings highlight the importance of understanding the expressiveness of the CMI framework. 

\section{Preliminaries}
\label{sec:preliminaries}
\newcommand{\rvar}{L}
\newcommand{\relm}{M}
\newcommand{\relmb}{N}

We consider the problem of binary classification,
with inputs in some space $\inspace$ assigned labels
in $\outspace = \{0,1\}$. 
A concept (or hypothesis) class $\HS\subseteq \outspace^ \inspace$ is a set of functions $h: \XX \to \YY$. We say $\HS$ \emph{shatters} $(x_1,\dots,x_m)\in \inspace^m$ if for all $(y_1,\dots,y_m)\in \{0,1\}^m$, there exists $h \in \HS$, such that, for all $i\in \range{m}$, we have $h(x_i)=y_i$. 
The \emph{VC dimension} of $\HS$, denoted by $\vcd$, is the largest $m \in \Naturals$ for which there exists $(x_1,\dots,x_m)\in \inspace^m$ shattered by $\HS$.
If no such finite $m$ exists, then $\vcd = \infty$. 

Let $\Dist$ be a distribution on 
$\dataspace =   \inspace \times \outspace  $. 
The \emph{empirical (classification) risk} of a classifier $h: \XX \to \YY$ on a sample $s=((x_1,y_1),\dots,(x_n,y_n)) \in \dataspace^n$
 is $\EmpRisk{s}{h} = n^{-1} \sum_{i \in [n]} \loss(h,(x_i,y_i))$, where
$\loss(h,(x,y)) = \indic{h(x) \neq y}$.
 Let $\SS \sim \Dist^n$, i.e., let $\SS$ be a sequence of i.i.d.\ random elements in $\dataspace$ with common distribution $\Dist$. (We can view $\SS$ itself as a random element in $\dataspace^n$.)
 The \emph{risk} of $h$ is $\Risk{\Dist}{h} = \EE\EmpRisk{\SS}{h}$, 
 where $\EE$ denotes the expectation operator. (The risk has, of course, no dependence on $n$ due to the data being i.i.d.) 

A distribution $\Dist$ is \defn{realizable by a class $\HS \subseteq \YY^\XX$} if 
there exists $h \in \HS$ such that $\Risk{\Dist}{h} = 0$. 
A sequence $((x_1,y_1),\dots,(x_n,y_n))$ is said to be \defn{realizable by
  $\HS$}, if for some $h \in \HS$, $h(x_i)=y_i$ for all $i \in \range{n} = \{1,\dots,n\}$. 
Note that if a distribution is realizable by  $\HS$, it implies that with probability one over $\SS \sim \Dist^n$, the training sample $\SS$ is realizable by $\HS$. 

Let $\TheAlg = (\Alg)_{n\ge 1}$ denote a (potentially randomized) learning algorithm, which, for any positive integer $n$, maps $\SS$ to an element of $\XX \to \YY$.
We say that $\TheAlg$ is a \emph{proper learner for a class $\HS \subseteq \XX \to \YY$} if the codomain of $\Alg$ is a subset of $\HS$ for every $n$. We say $\Alg$ is a consistent algorithm (learner) if $\EmpRisk{\SS}{\Alg(\SS)}=0$ a.s. Our primary interest in this paper is the \defn{expected generalization error} of $\Alg$ with respect to $\Dist$, defined as
$
\EGE = \EE \sbra[0]{\Risk{\Dist}{\Alg(\SS)} - \EmpRisk{\SS}{\Alg(\SS)} },$ where we average over both the choice of training sample and the randomness within the algorithm $\Alg$.

\subsection{Conditional mutual information (CMI) of an algorithm}
\label{subsec:cmi-defs}

In order to study generalization, and avoid some of the pitfalls of earlier
approaches based on mutual information, \citet{steinke2020reasoning} propose to
study the information contained in a ``supersample'' $\supersam$, a training sample
$\SS$ taken from the supersample, and the hypothesis $\Alg(\SS)$ output by a
possibly randomized learning algorithm, given $\SS$ as input.
Formally, let $\supersam=(Z_{i,j})_{i \in \{0,1\},\,  j \in \range{n}}$ to be an array of i.i.d.\ random elements in the space $\dataspace$ of labeled examples, with a common distribution $\Dist$.
In order to choose a training sample $\SS$ of size $n$ from $\supersam$,
let $U=(U_1,U_2,\dots,U_n)$ be a sequence of i.i.d.\ Bernoulli random variables in $\{0,1\}$, independent from $\supersam$, with $\Pr(U_i=0)=1/2$.
Define $\SS=\supersam_U=(Z_{U_j,j})_{j=1}^{n}$.
The \defn{conditional mutual information (CMI) of $\Alg$}, denoted $\SZB$,
is defined to be the mutual information between $\Alg(\SS)$ and $U$ given
$\supersam$,
denoted $I(\Alg(\SS);U|Z)$.  This quantity is equivalent to $I(\Alg(\SS);\SS|Z)$ when $\Dist$ is atomless, since $(U_1,\dots,U_n)$ is a.s.\ measurable with respect to $\SS$ and~$Z$. Because $Z$ and $U$ are independent, $\SZB\leq \entr{U \vert Z} = \entr{U}= n\log2$.
We now pause to introduce this and other
information-theoretic quantities formally.

\newcommand{\SSS}{\inspace}
\newcommand{\TTT}{\mathcal {T}}
\subsection{Measures of divergence
  and information}

Let $P,Q$ be probability measures on a measurable space. (We ignore
measure-theoretic pathologies for clarity.) For a $P$-integrable or nonnegative
function $f$, let $\smash{P[f] = \int f \dee P}$.
When $Q$ is absolutely continuous with respect to $P$, denoted $Q \ll P$,
write $\rnderiv{Q}{P}$ for (an arbitrary version of) the 
Radon--Nikodym derivative (or density) of $Q$ with respect to $P$. 
The \defn{KL divergence} (or \defn{relative entropy}) \defn{of $Q$ with respect to $P$},
denoted $\KL{Q}{P}$, is defined as $Q[ \log \rnderiv{Q}{P} ]$ when $Q \ll P$ and
infinity otherwise.

For a random element $X$ in some measurable space $\SSS$,
let $\Pr[X]$ denote its distribution, which lives in the space
$\ProbMeasures{\SSS}$ of all probability measures on $\SSS$.
Given another random element, say $Y$ in $\TTT$,
let
$\cPr{Y}{X}$ denote the conditional
distribution of $X$ given $Y$.
If $X$ and $Y$ are independent, $\cPr{Y}{X} = \Pr[X]$ a.s. 
For an event, say $X \in A$, 
$\cPr{Y}{X \in A}$ denotes the event's conditional probability
given $Y$, 
which is defined to be the conditional expectation of the indicator
random variable $\indic{X\in A}$ given $Y$,
denoted $\cEE{Y}{\indic{X \in A}}$.\footnote{
  By definition, $\cPr{Y}{X}$ 
  is a $\sigma(Y)$-measurable random element in $\ProbMeasures{\SSS}$, i.e., 
  $\cPr{\supersam}{U} = \kappa(\supersam)$ a.s.\ for some measurable map
  $\kappa : \TTT \to \ProbMeasures{\SSS}$.
  More generally, if, say $\cF=\sigma(Y,Z)$ is the $\sigma$-algebra generated by $Y$ and $Z$, then
  a conditional distribution/probability/expectation given $\cF$ is
  a measurable function of $Y$ and $Z$.
}
By the
chain (aka tower) rule, $\EE \cEE{\cF} = \EE$ for any $\sigma$-algebra $\cF$.

The \defn{mutual information between $X$ and $Y$} 
is $
\minf{X;Y} = \KL{\Pr[(X,Y)]} { \Pr[X] \otimes \Pr[Y]},$
where $\otimes$ forms the product measure. 
Writing $\cPr{Z}{(X,Y)}$ for the conditional distribution of the pair $(X,Y)$
given a random element $Z$,
the \defn{disintegrated mutual information between $X$ and $Y$ given $Z$,} is
\*[
 \dminf{Z}{X;Y} = \KL{\cPr{Z}{(X,Y)}}{\cPr{Z}{X} \otimes \cPr{Z}{Y} } .
\]
Then the \defn{conditional mutual information} of $X$ and $Y$ given $Z$ is $\minf{X,Y\vert Z} = \EE \dminf{Z}{X,Y}$.

Let $\mu = \Pr[X]$ and let $\kappa(Y) = \cPr{Y}{X}$ a.s. 
If $X$ concentrates on a countable set $V$ with counting measure $\nu$, 
the \defn{(Shannon) entropy of $X$} is 
$
\entr{X}=- \mu [ \log \rnderiv{\mu}{\nu}] = - \sum_{x\in V} \Pr(X=x)\,\log \Pr (X=x )
$. The \defn{disintegrated entropy of $X$ given $Y$} is defined by 
$
\centr{Y}{X} = -\kappa(Y)[ \log \rnderiv{\kappa(Y)}{\nu} ],
$
while the \defn{conditional entropy of $X$ given $Y$} is $\entr{X\vert Y} =  \EE[\centr{Y}{X}]$.
Note that $\entr{X\vert Y}\leq \entr{X}$. 
We will make use of the following lemma whose proof can be found in  \citep{madiman2007sandwich}.
\begin{lemma}
\label{lemma:chainrule-lowerbound}
Let $(X_1,X_2,\hdots,X_n)$ be  a discrete random vector, and $Y$ be an arbitrary random variable. Then, $
\entr{X_1,\hdots, X_n \vert Y} \geq \sum_{i=1}^{n} \entr{X_i \vert X_{-i},Y}$,
where $X_{-i}= (X_j: j \in \range{n}, j \neq i)$.
\end{lemma}

\citeauthor{steinke2020reasoning} establish a range of generalization bounds in terms of CMI. Our primary interest is in bounds for algorithms that have vanishing empirical risk. 
For $[0,1]$-bounded loss, \citeauthor{steinke2020reasoning}
show that 
\[
\label{eq:ege-interpolate}
 \EE \Risk{\Dist}{\Alg(\SS)} \leq 2 \EE \EmpRisk{S}{\Alg(\SS)} +\frac{3\SZB}{n }.
\]
For consistent learners (i.e., those that achieve zero empirical error a.s.), they also establish
\[
\label{eq:ege-fastrate}
\EE \Risk{\Dist}{\Alg(\SS)} &\leq \frac{\SZB}{n \log 2 }.
\]

\citeauthor{steinke2020reasoning} also introduce a variant of CMI based on the information revealed by the learner's losses on $\supersam$, rather than by the output hypothesis, $\Alg(\SS)$, directly.
\begin{definition}[Evaluated CMI, {\citep[][\S6.2.2]{steinke2020reasoning}}]
\label{def:ecmi}
Let $L\in \{0,1\}^{2\times n}$ be the array with entries $L_{i,j}=\loss(\Alg(\SS),Z_{i,j})$ for $i\in \{0,1\}$, $j\in \range{n}$.  The \defn{evaluated conditional mutual information of $\Alg$ with respect to $\Dist$},
denoted by $\eSZB$,
is the conditional mutual information $\minf{L ;U \vert \supersam}$. 
\end{definition}
By the data processing inequality, $\eSZB \leq \SZB$.  Therefore, $\eSZB$ is also bounded above by $n\log2$.  
For consistent learners, \citeauthor{steinke2020reasoning} show
 \[
 \EE \Risk{\Dist}{\Alg(\SS)} &\leq 1.5 \frac{ \eSZB}{n} \label{eq:ege-fastrate-ecmi}.
\]
For consistent learners $\Alg$ with bounded CMI or eCMI, 
these results imply their expected excess risk is of order $O(1/n)$.
The following result gives a nearly optimal bound
for the generalization error for VC classes in term of the evaluated CMI. The proof (\cref{pf:everyerm-log}) uses standard arguments, controlling the cardinality of the support of $L$ using the Sauer--Shelah lemma.
\begin{theorem}
\label{thm:everyerm-log}
For every $n$, let $\Alg : \dataspace^n \to \hypothesisclass_n$, where $\hypothesisclass_n$ is a concept class with VC dimension $\vcd_n$. 
Then, for every $n$ and distribution on $\supersam$, $\dminf{\supersam}{L;U} \le \vcd_n \log 6n$ a.s. 
In particular, $\sup_\Dist \eSZB \in O({\vcd_n \log n})$.
\end{theorem}

\begin{remark}
Markov's inequality and \cref{eq:ege-fastrate} imply $\Pr(\Risk{\Dist}{\Alg(\SS)} \geq \epsilon)\leq \SZB /( \log(2) n\epsilon) $ for consistent learners. By \citep[][Thm.~2.1]{haghifam2020sharpened}, 
$\minf{\Alg(\SS);U|Z} \le \minf{\Alg(\SS);\SS}$.
This observation, combined with \citep[][Prop.~11]{bassily2018learners}, implies
there is an input space, data distribution, and consistent learning algorithm for which this tail bound's dependence on $n$ is \emph{tight}. 
If one were to obtain sample complexity bounds via such tail bounds, one would only prove that $O(1/(\epsilon\delta))$ samples suffice to find a hypothesis with $\epsilon$ estimation error with probability at least $1-\delta$. 
The linear dependence on $1/\delta$ is, however, suboptimal.
As such, it seems that the CMI framework cannot be used to obtain optimal sample complexity bounds in the PAC framework. Recent proposals for disintegrated notions of CMI in \citep{hellstrom2020generalization} might
provide a framework for studying the sample complexity of PAC learning using an information-theoretic framework. %
\end{remark}

\section{Optimal CMI Bound for SVM and Stable Compression Schemes}
\label{sec:stable-compression}

In this section, we show that the CMI framework can be used to derive an optimal
excess risk bound for the SVM algorithm learning half spaces in $\Reals^d$. 
To show this, we establish optimal CMI bounds for the subclass of \emph{stable}
sample compression schemes,
which imply this section's main result:
\begin{theorem}
\label{thm:svm}
Let $\Alg$ be the SVM algorithm for learning the class of half spaces in $\Reals^d$. Then, for every $n > {d}/{2}$ and realizable distribution $\Dist$ in $\Reals^d$, we have $\SZB \leq 2(d+1)\log 2$. 
\end{theorem}
Combining this result with \cref{eq:ege-fastrate} gives $ \EE  \Risk{\Dist}{\Alg(\SS)}\leq 2(d+1)/n$. The lower bound for expected excess risk of linear classifiers in  \citep{long2020complexity} shows this bound is optimal up to a constant factor.

\subsection{CMI of Stable Compression Schemes}

\citet{littlestone1986relating} introduced compression schemes, which
capture the idea that a consistent hypothesis can be defined in terms of a fixed number of samples.
Formally, for a concept class 
$\hypothesisclass \subseteq \outspace^\inspace$,
a \defn{sample compression scheme} of size $k \in \Nats$ is a pair $(\kappa,\rho)$ of maps
such that, for all 
samples $s = ((x_i,y_i))_{i=1}^{n}$ of size $n \ge k$,
the map $\kappa$ compresses the sample into a length-$k$ subsequence $\kappa(s) \subseteq s$
which the map $\rho$ uses to reconstruct an empirical risk minimizer 
 $\hat{h}=\rho(\kappa(s)) $. %
\citeauthor{steinke2020reasoning}  prove the following upper bound on the CMI of a sample compression scheme.

\begin{theorem}[{\citealp[Thm.~4.1]{steinke2020reasoning}}]
\label{thm:cmi-general-compression}
Let $\hypothesisclass$ be a hypothesis class
that has a sample compression scheme $(\kappa,\rho)$  of size $k$. Then, $\SZB \leq k \log(2n)$ 
where $\Alg(\cdot) = \rho(\kappa(\cdot))$.
\end{theorem}
Note that the bound in \cref{thm:cmi-general-compression} \emph{cannot} be improved from $O(k \log n)$ to $O(k)$ for \emph{every} sample compression scheme, and so the bound in \cref{thm:cmi-general-compression} is tight, and cannot be improved without further information about the compression scheme. 
The proof of the optimally stems from the fact that there exists compression schemes of size $k$ and data distributions $\Dist$ such that there is a lower bound $\mathbb{E}[R_{\Dist}(\mathcal{A}_n)]=\Omega({k \log (n)}/{n})$ where  $\mathcal{A}_n (\cdot)=\rho(\kappa(\cdot))$ \citep{hanneke2019sharp,floyd1995sample}. Combining this lower bound with \cref{eq:ege-fastrate} proves the optimally of \cref{thm:cmi-general-compression}.

Nevertheless, we can  circumvent this lower bound by considering an important subclass of the sample compression schemes. Many natural compression schemes are also \emph{stable}
in the sense that removing any training example that was not in the compressed sequence
does not alter the resulting classifier. To give a formal definition, we write $s \subseteq s'$ for two sequences $s,s'$ if, under some permutation, $s$ is a subsequence of $s'$.

\begin{definition}[Stable sample compression scheme; \citealt{bousquet2020proper}]
A sample compression scheme $(\kappa,\rho)$ of size $k$ is said to be \defn{stable} 
if $\kappa$ is symmetric (i.e., invariant to permutation of its input) and,
for every realizable sample $s$ of size $n \ge k$,
and every sequence $s'$ such that $\kappa(s) \subseteq s' \subseteq s$, we have $\rho(\kappa(s)) = \rho(\kappa(s'))$.
Due to the symmetry of $\kappa$, we refer to its output as the compression \defn{set}, although
the equivalence class of sequence under permutations is the structure of a multiset, not a set.
\end{definition}

The concept of a stable compression scheme has its roots in the analysis of the SVM for learning half-spaces in $\Reals^d$ \citep{vapnik1974theory},
which is the quintessential example of a stable sample compression scheme.
For SVMs, the compression (multi)set contains at most $d+1$ distinct ``support vectors'' for any given training set. 
The reconstruction map outputs the max-margin classifier over the set of support vectors. 
By stability, removing any training example that is not a support vector does not change the resulting classifier \citep[][Sec.~5.3.2]{mohri2018foundations}. 
In the next theorem, we present a uniform CMI bound over realizable distributions for every stable sample compression scheme. Our bound removes the 
$\log n $ factor from \cref{thm:cmi-general-compression} and is \emph{optimal} up to a constant factor in the distribution-free setting.
\begin{theorem}
\label{thm:stable-compression-cmi}
Let $\hypothesisclass$ be a concept class with a stable compression scheme $(\kappa,\rho)$ of size $k$. Then, for  every realizable data distribution $\Dist$ and $n \geq k$, $\SZB \leq  2k \log 2$,
where $\Alg = \rho(\kappa(\cdot))$.
\end{theorem}
\begin{remark}
 \citet[][Sec. 4.4]{steinke2020reasoning} propose an algorithm for learning threshold functions (positive rays) in the realizable setting  over $\Reals$ that achieves $\SZB\leq 2\log2$.  
 It is interesting to note that their algorithm can be viewed as a stable compression scheme. Specifically, for a realizable training set $s$, let $x^\star=\min\{x \in \Reals: (x,1)\in s\}$ if $s$ has any sample with label $1$, otherwise let $x^\star = \infty$. 
 Then the algorithm proposed by \citeauthor{steinke2020reasoning} is $\Alg(\SS)=\hat{h}$, where $\hat{h}(x) = \indic{x\geq x^\star}$.  \citeauthor{steinke2020reasoning} present a bespoke analysis of this special algorithm.
 It is straightforward to see that the algorithm is a stable compression scheme of size one and the compression map here is symmetric. 
 Therefore, the result of \cref{thm:stable-compression-cmi} gives $\SZB\leq 2\log 2$. 
\end{remark}
\begin{proof}[Proof of \cref{thm:stable-compression-cmi}]
Let $W=\Alg(Z_U)=\rho(\kappa(Z_U))$ and note that $\Alg$ is deterministic.   We have
  $\entr{U | \supersam} = n \log 2$ due to independence of $U$ and $\supersam$
  and the independence of components of $U$.
  Then, by the definition of mutual information in terms of entropy, and \cref{lemma:chainrule-lowerbound}, 
\begin{align}
\label{eq:lemma-lowerbound-use}
\SZB &=  \entr{U | \supersam} - \entr{U | W, \supersam}
\leq n \log 2 - \sum_{i=1}^{n} \entr{U_i| W,\supersam, U_{-i}}.
\end{align}
Fix $i \in \range{n}$, and define $\Ur{i}{b}\triangleq
(U_1,\dots,U_{i-1},b,U_{i+1},\dots,U_n)$ for $b \in \{0,1\}$. Using this
notation, we can define two training sets  $\Sr{i}{b}= \supersam_{\Ur{i}{b}}$ for
$b \in \{0,1\}$. Let $\cF_i$ be the $\sigma$-algebra $\sigma(W,\supersam,U_{-i})$ and let $E$ be the event
$\rho(\kappa(\Sr{i}{0}))= \rho(\kappa(\Sr{i}{1}))$. Then, 
by the non-negativity of entropy,  
\[
\label{eq:decompose-entropy}
\entr{U_i| W,\supersam, U_{-i}}
= \EE \big[ \centr{\cF_i }{U_i} \big ] \ge \EE \big[ \centr{\cF_i }{U_i} \indic{E}
\big].
\]
Note that, conditional on the sub-$\sigma$-algebra $\cG_i = \sigma(\supersam,U_{-i})$,
$W$ takes on at most two values.
However, on the event $E$ (or equivalently, conditioning further on the event $E$, since $E$ is $\cG_i$-measurable), 
$W$ is now nonrandom because it takes on a single value.
It follows that, conditional on $\cG_i$ and the event $E$,
$W$ is trivially independent of every random variable, including $U_i$.
Ergo, on $E$, 
$\cEE{\cG_i}[U_i] = \cEE{\cF_i}[U_i] = \cPr{\cF_i}{U_i=1}$.
But $U_i$ is independent of $\cG_i$,
and so $\cEE{\cG_i}[U_i] = \EE[U_i] =  \frac 1 2$. 
Thus, on $E$,
$\cPr{\cF_i}{U_i=1} =  \frac 1 2$ and so $\centr{\cF_i}{U_i} = \binaryentr{\frac 1 2} = \log 2$.
Therefore,
\[
  \entr{U_i| W,\supersam, U_{-i}}
  \ge
  \log 2 \cdot \Pr\big( \rho(\kappa(\Sr{i}{0}))=\rho(\kappa(\Sr{i}{1}))  \big) . \label{eq:lower-bound-prob}
\]
We can bound the probability of $E$ from below
using the stability property of the compression scheme. For any $(x_1,\tilde{x}_1,x_2,\dots,x_{n})\in \inspace^{n+1}$ and $h \in \hypothesisclass$, consider two multisets $S=\{ (x_1,h(x_1)), (x_2,h(x_2)),\hdots, (x_{n},h(x_{n}))\} $ and  $\tilde{S}=\{ (\tilde{x}_1,h( \tilde{ x}_1)),$ $(x_2,h(x_2)),\hdots, (x_n,h(x_n))\}$, where  $S$ and $\tilde{S}$ differ only in the first element. Define the multiset $S\cup \tilde{S}=\big\{ (x_1,h( x_1)), (\tilde{x}_1,h( \tilde{ x}_1)),$ $(x_2,h(x_2)),\hdots, (x_n,h(x_n))\big\}$. 
We claim that if $ (x_1,h( x_1))$ and $(\tilde{x}_1,h( \tilde{ x}_1))$ are not the members of the compression set $S\cup \tilde{S}$, then $ (x_1,h(x_1))$ and $ (\tilde{x}_1,h( \tilde{ x}_1))$ are not in the compression set of $S$ and $\tilde{S}$, respectively. 
To prove this claim, since  $(x_1,h(x_1))$ is not in the compression set  $S\cup \tilde{S}$ by the stability of $\kappa$, we have $\rho(\kappa(S\cup \tilde{S}))=\rho(\kappa(S\cup \tilde{S} \setminus \{(x_1,h(x_1))\}))$. 
By the definition of $S$ and $\tilde{S}$, $S\cup \tilde{S} \setminus \{(x_1,h(x_1))\} = \tilde{S}$. 
Thus, combining facts that  $(x_1,h(x_1))$ is not in the compression set $S\cup \tilde{S}$  and $\kappa(S \cup \tilde{S})=\kappa(S \cup \tilde{S} \setminus \{(x_1,h(x_1))\})=\kappa(\tilde{S})$, we obtain $(\tilde{x}_1,h( \tilde{ x}_1))$ is not in the compression set $\tilde{S}$. Similarly, we can prove  $(x_1,h( x_1))$ is not a member of the compression set $S$ by switching $x_1$ with $\tilde{x}_1$ in the argument. 
By this argument,
\[
&\Pr\big( \rho(\kappa(\Sr{i}{0}))=\rho(\kappa(\Sr{i}{1}))  \big) \geq \Pr\big( Z_{0,i}\not\in \kappa(\Sr{i}{0}\cup \Sr{i}{1})\wedge Z_{1,i} \not\in \kappa(\Sr{i}{0}\cup \Sr{i}{1})  \big)\label{eq:fact-stability}.
\]
Recall that the elements of $\supersam$ are \iid, hence exchangeable. Since the size of the sample compression is $k$ and $\kappa$ is symmetric, we have 
\begin{align}
\label{eq:prob-union-compressionset}
 \Pr\big( Z_{0,i}\not\in \kappa(\Sr{i}{0}\cup \Sr{i}{1})\wedge Z_{1,i} \not\in \kappa(\Sr{i}{0}\cup \Sr{i}{1})  \big)
  \geq  \textstyle \binom{n-1}{k} / \binom{n+1}{k} .
\end{align}
Combining \cref{eq:lower-bound-prob,eq:fact-stability,eq:prob-union-compressionset} yields
$\entr{U_i| W,\supersam, U_{-i}} \geq \log 2 \cdot \binom{n-1}{k} / \binom{n+1}{k} \geq(1-{2k}/{n}) \log2  $.
Finally, the result follows by substitution of this bound into \cref{eq:lemma-lowerbound-use}.
\end{proof}
The result for SVMs (\cref{thm:svm}) follows immediately from \cref{thm:stable-compression-cmi} and the fact 
that the SVM may be expressed as a stable compression scheme of size $d + 1$.

\section{CMI of Proper Learning of VC classes}
\label{sec:proper}

Following their paper introducing CMI, \citeauthor{steinke2020openproblem} posed several open problems 
asking whether VC classes under realizibility admit learners with bounded CMI.
We will restate their conjectures, and then showing that there exist some VC classes for which it is not possible to find a proper learner with bounded CMI under realiziblity. 
We then consider a subset of VC classes, namely VC classes with finite star number, and show that for such concept classes, there exists an ERM with bounded CMI.

We first state the main result of \citep{steinke2020reasoning} on the CMI of proper learners.
\begin{theorem}[Thm.~4.12, \citealt{steinke2020reasoning} ]
\label{thm:sz-thm-proper}
Let $\hypothesisclass$ be a concept class with VC dimension $\vcd$. Then for all $n \in \Nats$, there exists a proper ERM algorithm $\Alg$ for learning $\hypothesisclass$ such that for every realizable distribution $\Dist$, $\SZB=O\big(\vcd \log n\big)$.
\end{theorem}
\begin{remark}[Comparison of \cref{thm:sz-thm-proper} and \cref{thm:everyerm-log}]
First, note that \cref{thm:sz-thm-proper}
does not hold for every ERM algorithm. 
As discussed in \citep{steinke2020reasoning}, we can construct pathological ERMs with nearly maximal CMI by simply encoding the information $U$ into the ``lower-order'' bits of $W$. 

It is also worth noting that our result in \cref{thm:everyerm-log} is more general. 
There we show that a  bound $O\big(\vcd \log n\big)$ holds for \emph{evaluated} CMI of \emph{any algorithm} that outputs a hypothesis from VC class, whereas \cref{thm:sz-thm-proper} holds for a \emph{specific} proper algorithm.  
\end{remark}

\subsection{A Limitation of Proper Learning}
\citet{steinke2020openproblem} propose two conjectures regarding CMI for proper learning of VC classes under the realizability assumption, both of which can be seen as special cases of the following statement:

\begin{statement}
\label{conjecture:realizable-generalf}
There exists a real-valued function $f$ and constant $c \ge 0$ such that, 
for every nonnegative integer $\vcd$ and VC class $\hypothesisclass \subseteq \XX \to \YY$ of dimension $\vcd$,
there exists a proper learning algorithm $\TheAlg$ for $\HS$ 
such that,
for every $n \ge \vcd$, 
$\SZB \leq f(\vcd)$ for all $\Dist$
and, for every realizable $s \in \dataspace^{n}$, 
$\EE \EmpRisk{s}{\Alg(s)}  \leq c\,\vcd / n$,
where the expectation is taken only over the randomness in $\Alg$.
\end{statement}

\citet{steinke2020openproblem} conjecture that \cref{conjecture:realizable-generalf} holds for $f$ linear.
In this section, we show that
\cref{conjecture:realizable-generalf} is false in general: it is not possible to find a proper learning algorithm for \emph{every} VC class that removes the $\log(n)$ factor from \cref{thm:sz-thm-proper}.
For a class $\HS \subseteq \XX \to \YY$, let $\samplecomplex{\epsilon}{\delta}$ denote the \defn{proper optimal sample complexity} of $(\epsilon,\delta)$-PAC learning $\HS$, 
i.e., $\samplecomplex{\epsilon}{\delta}$ is the least integer $n$, for which there exists a proper learning algorithm $\TheAlg$ such that, for every realizable distribution $\Dist$, $
\Pr(\Risk{\Dist}{\Alg(\SS)} \geq \epsilon) \leq \delta.$
The following result provides a lower-bound on the sample complexity of proper learning:
\begin{theorem}[Thm.~11, \citealt{bousquet2020proper}] 
\label{thm:sample-complexity-lowerbound}
Let $\epsilon \in (0,1 / 8)$ and $\delta \in (0,1/100)$. There exists a concept class with VC dimension $\vcd$.  for which we have $\samplecomplex{\epsilon}{\delta} \geq \frac{\tilde{c}}{\epsilon}(\vcd\, \mathrm{Log} \frac{1}{\epsilon} + \mathrm{Log}\frac{1}{\delta})$ for a fixed numerical constant $\tilde{c} > 0$, where $\mathrm{Log}(x) = \max \{1, \log(x) \}$ for $x\geq 0$.
\end{theorem}
We now present the main result:
for VC classes, we show that the existence of a learning algorithm with bounded CMI contradicts the lower bound on the sample complexity in \cref{thm:sample-complexity-lowerbound}. 
The proof can be found in \cref{pf:sample-complexity-contradiction}.

\begin{theorem}\label{mainthm}
\cref{conjecture:realizable-generalf} is false.
\end{theorem}
\begin{remark}
Consider a modified \cref{conjecture:realizable-generalf}, seeking a proper learner with bounded eCMI instead. We can show that this modified statement is also false. 
\end{remark}

\subsection{VC Classes with Finite Star Number}
\label{sec:star-number}

\cref{mainthm} states that it is not possible to find a proper learning algorithm with bounded CMI for \emph{every} VC class. 
Note that this limitation does not imply a failure of the CMI framework for characterizing the expected excess risk of learning VC classes. Instead, the impossibility can be attributed to an inherent limitation of proper learning algorithms, since there exist VC classes such that no proper learning algorithm $\Alg$ satisfies $\EE[\Risk{\Dist}{\Alg}]=O({1}/{n})$ \citep{bousquet2020proper}. 
In this section, we consider a family of VC classes for which we \emph{can} show the existence of a learner with bounded CMI. We begin with some definitions.

Two sequences $((x_1,y_1),\dots,(x_n,y_n))$ and $((x'_1,y'_1),\dots,(x'_n,y'_n))$ are \defn{neighbours} if $x_i=x'_i$ for all $i \in \range{n}$, and $y_i = y'_i$ for all but exactly one $i \in \range{n}$. 
Fix any concept class $\hypothesisclass \subseteq \outspace^{\inspace}$. \defn{Star number of $\hypothesisclass$} \citep[][Def.~2]{hanneke2015minimax},  denoted by $\starnum$,
is the largest integer $n$ such that there exists a realizable $s \in (\inspace \times \outspace)^n$,  and every neighbour of $s$ is realizable by $\hypothesisclass$. If no such largest integer $n$ exists, then $\starnum=\infty$.  
\citet[Sec. 4.1]{hanneke2015minimax} calculate the star
number of some common concept classes. 
It is straightforward to see that $\vcd \leq \starnum$. For any $n \in \Naturals$, and  $s=((x_1,y_1),\dots,(x_n,y_n))\in (\inspace \times \outspace)^n$, define a \emph{version space} of $s$ with respect to $\hypothesisclass$  as  $\versionspace{s}=\{ h\in \hypothesisclass: \EmpRisk{s}{h}=0 \}$, a set of classifiers that are consistent with $s$.

\subsubsection{Star Number, Version Space, and CMI}

Fix any concept class $\hypothesisclass $, %
and assume that, after observing a training sample $\SS$, 
we want to output the version space $\versionspace{\SS}$, i.e., the set of all classifiers consistent with $\SS$. 
We are interested in the following question: for which concept classes does the version space carry little information about the training samples conditioned on the supersample? More precisely, for which classes is $\minf{\versionspace{\SS};U\vert Z}=O(1)$? 
Note that bounding the ``CMI'' of the version space provides a bound on the CMI of a broad class of algorithms that choose a particular ERM based solely on the version space, potentially under further constraints, such as privacy, fairness, etc.

In this section, we give a complete characterization of when $\minf{\versionspace{\SS};U\vert Z}=O(1)$, and show that it is possible if and only if $\HS$ has finite star number. 
In particular, given a class with infinite star number, 
we demonstrate that  $\minf{\versionspace{\SS};U\vert Z}=\Omega(n)$. 
We begin with an upper bound, whose proof can be found in \cref{pf:publish-versi}. 
\begin{theorem}
\label{thm:publish-versionspace}
Let $n \in \Naturals$, $\hypothesisclass$ be a concept class with  star number $\starnum$, and  $\Dist$ be a realizable distribution. Let $\supersam$, $U$, and  $\SS$ be as defined in the beginning of this section.  Then for every $n\geq \starnum$, we have $\minf{\versionspace{\SS};U\vert Z} \leq 2 \starnum \log 2 $. 
\end{theorem}
We can use the data processing inequality and \cref{thm:publish-versionspace} to obtain the following: 
\begin{corollary}
\label{thm:etd}
Let $\hypothesisclass$ be a concept class with the star  number $\starnum$.  Consider \emph{any} ERM algorithm $\Alg$ for which the Markov chain $\SS-\versionspace{\SS}-\Alg(\SS)$ holds; in other words, the output of the algorithm and the training set are conditionally independent given the version space. Then, for any such an algorithm, for every $n\geq \starnum$, and every realizable distribution $\Dist$, we have 
$ \SZB  \leq 2 \starnum \log 2 .$
\end{corollary}
In \cref{thm:etd},  by assuming the Markov structure $\SS-\versionspace{\SS}-\Alg(\SS)$  we restrict the information of the ERM algorithm $\Alg(\SS)$. One might try to extend our result in \cref{thm:etd} such that it holds for \emph{any} ERM without any constraints. However, for the class of one-dimensional threshold over $\Reals$, whose star number is two, one can construct an ERM  with maximal CMI \citep[][Sec.~4.3]{steinke2020reasoning}. Therefore, the Markov chain assumption cannot be removed.  
The next theorem shows $\starnum<\infty$ is a necessary condition, for otherwise, there exist learning scenarios under which we cannot output the version space, even with merely sublinear CMI.
\begin{theorem}
\label{thm:conv-publish-versionspace}
For every $n\in \Naturals$, $n\geq 2 $ and for every concept class $\HS$ with star number $\starnum$ with $\starnum\geq 2$  over input space $\inspace$, there exists a realizable data distribution $\Dist$ on $\inspace \times \outspace$ such that 
$\minf{\versionspace{\SS};U\vert Z} = \Omega(\min\{\starnum,n\})$.
\end{theorem}
\begin{proof}[Proof sketch:]
Let $\inspace=[n]$
and consider the concept class $\HS = \{ h_0,h_1,\dots,h_n : \inspace \to \outspace \}$,
where $h_0(x)=0$ is the zero function and $h_t(x)=\indic{x=t}$, for $t \in [n]$, are point functions.
It is easy to see that this concept class has star number $n$ on $\inspace$. 
Let $\Dist$ correspond to the uniform distribution on $\inspace$ and target function $h_0$.
Consider the bijection between $\HS$ and $\{0,1,\dots,n\} \supseteq \inspace$.
For every training sequence, the version space contains 0 and every point in $\inspace$ not observed in $\SS=Z_U$. The key observation is that, in each column of $Z$, one point was \emph{not} selected for training, and so each column contains zero or one points in the version space. Whenever there is one point, the value of $U_i$ is revealed for that column. We show that the number of columns with this property is a lower bound on $\minf{\versionspace{\SS};U\vert Z}$. A coupon collector’s argument yields a lower bound the number of such columns. The formal proof can be found in
\cref{pf:conv-publish-versionspace}.
\end{proof}

\subsubsection{An ERM whose CMI is logarithmic in  star number}
In the next theorem, we show that there exists an ERM for learning VC classes with a finite star number for which the CMI is upper bounded by a constant and its dependence on star number is logarithmic. The proof is provided in \cref{pf:starnumber-log}.
\begin{theorem}
\label{thm:starnumber-log}
Let $\hypothesisclass$ be a concept class with VC dimension $\vcd$ and star number $\starnum$. Then, there exists an ERM $\Alg$ for learning $\hypothesisclass$ such that for every $n \geq \starnum$ and for every realizable distribution $\Dist$, we have $\SZB=O\big(\vcd \log ( \starnum /\vcd  )\big)$.
\end{theorem}
Note that \cref{thm:starnumber-log} shows the existence of a specific ERM with constant CMI, whereas in \cref{thm:etd} we show a broad class of ERMs has bounded CMI. 
\section{Universality of eCMI and Improper Learning of VC Classes}

The eCMI, introduced in \cref{def:ecmi}, is an appropriate information-theoretic notion for analyzing learning algorithms when there is no natural parameterization of the set of possible predictors, such as for improper or transductive algorithms. In this section, we show that eCMI is \emph{universal} in the realizable setting. 
Then, we show that the CMI framework can be used to obtain a near-optimal bound on the expected excess risk of any algorithm with a leave-one-out error guarantee. As an application, we study CMI of the classical \emph{one-inclusion graph prediction} algorithm, which was first  proposed by \citet{haussler1994predicting} as an optimal improper learner for VC classes. The next theorem is the main result of this section, whose proof can be found in \cref{pf:ecmi-every-consistent}.
\begin{theorem}
\label{thm:ecmi-consistent}
Let $n \geq 2 \in \Naturals$, let $\Alg$ be a learning algorithm, and let $\Dist$ be a distribution on $\dataspace$. Assume with probability one $\EmpRisk{\SS}{\Alg(\SS)}=0$.  Then,  
\[
\label{eq:ecmi-any-consistent}
2/3 \Risk{\Dist}{\Alg} \stackrel{(\text{a})}{\leq} \eSZB/n  \stackrel{(b)}{\leq}  \binaryentr{\Risk{\Dist}{\Alg}} + \Risk{\Dist}{\Alg} \log(2),
\]    
where $\binaryentr{\cdot}$ is the binary entropy function, and $\Risk{\Dist}{\Alg}=\EE[\Risk{\Dist}{\Alg(\SS)}]$. 
\end{theorem}
The inequality $(a)$ in \cref{eq:ecmi-any-consistent} implies that, if $\eSZB/n$  vanishes as $n$ diver2ges, then  $\Risk{\Dist}{\Alg}$ vanishes as well. The inequality $(b)$ is more interesting: it implies that, if $\Risk{\Dist}{\Alg}$  vanishes as $n$ diverges, then $\eSZB/n $  also vanishes.
\color{black}

Assume that a consistent algorithm $\TheAlg$ satisfies $\Risk{\Dist}{\Alg} = \theta / n$ for $\theta \in \Reals \geq 1 $. Then, it is straightforward to see from Direction $(b)$ in \cref{eq:ecmi-any-consistent} that  $\eSZB/ n = O(\theta \log(n))$. Also, for an algorithm with  $\Risk{\Dist}{\Alg} = \theta \log(n)/ n$ the upper bound in \cref{eq:ecmi-any-consistent} is given by $O(\theta (\log(n))^2)$. This observation suggests that our upper bound for  $\eSZB/ n$ in \cref{eq:ecmi-any-consistent} provides a bound on the expected excess risk which is sub-optimal by a $\log(n)$ factor in some interesting cases. 
\begin{remark}
Note that the result in \cref{thm:ecmi-consistent} \emph{does not} imply our results in former sections. In particular our results in \cref{thm:everyerm-log}, \cref{thm:stable-compression-cmi}, \cref{thm:etd}, and (later in) \cref{thm:singelton} show that CMI framework provides \emph{optimal} characterization of the expected excess risk in the considered scenarios.
\end{remark}
The following corollary summarizes our result for the consistent algorithms with a leave-one-out error guarantee. 
\begin{corollary}
\label{cor:loo-ecmi}
Let $n \in \Naturals$ and $\theta \in \Reals_{+}$, such that $n\geq 2\theta$. 
Let $\Alg$ be a consistent learning algorithm.
Let $\Dist$ be a distribution on $\dataspace$ and assume that, with probability one 
over a sequence $S = (Z_1,\dots,Z_{n+1})\dist\Dist^{n+1}$, 
we have $\frac{1}{n+1}\sum_{i=1}^{n+1}  \cEE{S}[\loss(\Alg(S_{-i}),Z_i)] \leq \frac{\theta}{n+1}$, where the expectation is taken only over the randomness in $\Alg$. Then, 
\*[
\eSZB \leq \theta \log( {(n+1)}/{\theta}) + 2\theta\log 2   .
\]    
\end{corollary}

\newcommand{\probassignment}{f_{\oneinclusiongraph(\bar{X})}}
\subsection{The One-Inclusion Graph Prediction Strategy}
\citet{haussler1994predicting} proposed  an improper learning rule for learning VC classes based on the one-inclusion graph \citep{alon1987partitioning}. We  provide a description of this algorithm in \cref{app:description-oneinclusion}. The deterministic version of this prediction rule satisfies the following property. Let $\HS$ be a concept class with VC dimension $\vcd$. For every $n \in \Naturals$, $h \in \HS$, and $(x_1,\dots,x_{n+1})\in \inspace^{n+1}$, let $S=((x_1,h(x_1)),\dots,(x_{n+1},h(x_{n+1})))$. Then  $\frac{1}{n+1}\sum_{i=1}^{n+1}  \loss(\Alg(S_{-i}),(x_i,h(x_i))) \leq \frac{\vcd}{n+1}$.  A direct application of \cref{cor:loo-ecmi} gives the following results.
\begin{corollary}
\label{corr:one-inclusion-determinstic}
Let $\Alg$ denote the deterministic one-inclusion graph for learning class $\HS$ with VC dimension $\vcd$. Then, for every realizable distribution $\Dist$ and $n\geq 2\vcd$, we have  $\eSZB\leq \vcd \log( {(n+1)}/{\vcd}) + 2 \vcd \log2  $.
\end{corollary}
\begin{remark}
In \cref{thm:everyerm-log} we provide a bound on eCMI of any proper ERM. However, for improper learners, we can construct a consistent algorithm with maximal eCMI. 
For instance, consider $\inspace=[0,1]$, $\Dist_X=\unif{[0,1]}$, the concept class of threshold with target function $\targetfun(x)=\indic{x\geq 1/2}$. Consider a learning algorithm that gives the correct predictions on the points that are in the training set, and for a point that is not in the training set it always predicts one. One can show that eCMI of this consistent algorithm is $ \Omega(n)$. 
\end{remark}
\citet{haussler1994predicting} showed that the one-inclusion graph algorithm achieves $\EE \Risk{\Dist}{\Alg(\SS)}\leq \vcd / n$ for learning a class $\HS$ with VC dimension $\vcd$. \cref{corr:one-inclusion-determinstic} implies that $\eSZB = O(\vcd \log(n))$ for every \emph{deterministic} one-inclusion graph prediction rule. 
Combining this result with \cref{eq:ege-fastrate-ecmi} provides a bound on the excess risk which is suboptimal by a $\log n$ factor. In the next theorem, we show that, in at least one interesting special case, it is possible to remove
the logarithmic factor from eCMI by exploiting a \emph{randomized} one-inclusion graph prediction algorithm. 
\begin{theorem}
\label{thm:singelton}
Let $\HS$ denote the class of singletons (point functions) on $\inspace = \Reals$. There exists a randomized one-inclusion graph prediction rule $\Alg$ for learning class $\HS$ such that for every realizable distribution $\Dist$ and $n\geq 2$, we have  $\eSZB=O(1)$.
\end{theorem}

\color{black}
\section{Remaining Gaps and Open Questions}
\label{sec:openq}
For proper learning of VC classes, \citet{hanneke2016refined}  showed the assumption $\starnum < \infty$ is a necessary and sufficient condition for the existence of a distribution-free bound on the expected risk of all ERMs converging at a rate $O(1/ n)$. In \cref{thm:etd} and \cref{thm:starnumber-log} , we showed  the same rate for the expected risk of a broad class of ERMs can be obtained using the CMI framework. It is an open question to show that for a class with finite star number, every ERM has bounded eCMI.

An important open problem is to show that for every VC class with finite dual Helly number \citep{bousquet2020proper} there exists a proper learning algorithm such that for every data distribution its expected empirical risk converges at a rate of $O(1/n)$ and it has bounded CMI. Combining the generalization guarantees that one can retrieve from \cref{eq:ege-fastrate} the expected excess risk of the learner with these properties matches the optimal rate from \citet{bousquet2020proper}.

For improper learning of VC classes, we showed a general result for the deterministic one-inclusion graph prediction rule which is suboptimal by a $\log n$ factor. We conjecture that for every VC class with dimension $\vcd$ there exists a probability assignment for the randomized one-inclusion graph for which eCMI is $O(\vcd)$. In \cref{thm:singelton}, we showed this claim holds for the class of point functions.

We also remark that if the answers to the above questions are affirmative, then it can be argued that the CMI framework is expressive enough so that it can explain generalization properties of VC classes. Otherwise, a negative answer to any of the questions implies that there is gap  between CMI framework and VC theory.

In \cref{thm:ecmi-consistent} we proved that eCMI is \emph{universal} in the realizable setting. A fundamental question to ask is whether for every data distribution  $\Dist$ and consistent learner $\TheAlg$, $\eSZB/ n$ vanishes as the number training samples $n$ diverges \emph{at the same rate} with the excess risk, i.e., $\Risk{\Dist}{\Alg}$.

\subsection*{Acknowledgments}
The authors would like to thank Blair Bilodeau, Mufan Bill Li, and Jeffery Negrea for feedback on drafts of this work.

\subsection*{Funding}
M.\ Haghifam is  supported by the Vector Institute, University of Toronto, and a MITACS Accelerate Fellowship with Element AI. S.\ Moran is a Robert J.\ Shillman Fellow and is supported by the ISF, grant no.~1225/20, by an Azrieli Faculty Fellowship, and by BSF grant 2018385. D. M. Roy was supported, in part, by an NSERC Discovery Grant, Ontario
Early Researcher Award, and a stipend provided by the Charles Simonyi Endowment.
Resources used in preparing this research were provided, in part, by the Province of Ontario, the Government of Canada through CIFAR, and companies sponsoring the Vector Institute \url{www.vectorinstitute.ai/partners}.
\color{black}

\printbibliography

\newpage
\appendix{}
\section{Known Bounds for Learning VC Classes}
\label{apx:known-bounds}
In this part, we provide a landscape of the known results for the learning VC classes. One key distinction is proper learning versus improper learning. In particular, for every VC class with dimension $d$, there exists a consistent and improper learning algorithm that achieves $O(d/n)$ risk under realizability, and this bound is optimal \citep{hanneke2016optimal,haussler1994predicting}. The situation for proper learning is much more complicated. In general, the achievable rate for the proper learning of VC classes is off by a log factor, i.e., $O(d\log(n)/n)$. \citet{bousquet2020proper} show that when the dual Helly and hollow star number, which are combinatorial complexity measures of the class, agree, then they characterize the existence of an optimal proper learner. Also, a subclass of VC for which the log factor provably cannot be removed using proper learners is characterized in \citep[][Thm.~11]{bousquet2020proper}.  Moreover, for general ERMs, \citet{hanneke2016refined} shows the finiteness of  star number is a necessary and sufficient condition under which we can remove the log factor using any arbitrary ERMs. 

It is interesting to note that the results of \citet{moran2016sample} reveal a connection between the general sample compression schemes and VC classes. However, it is not known whether, in general, the optimal rates for VC classes are always witnessed by compression schemes.

\section{Proof of \cref{thm:everyerm-log}}
\label{pf:everyerm-log}

Fix $n \in \Naturals$.
We have $\eSZB = \EE[ \dminf{\supersam}{L ; U}]$ by definition,
and $\dminf{\supersam}{L;U} \le \centr{\supersam}{L}$ a.s., by the nonnegativity of (conditional) entropy.
It then suffices to bound the cardinality, $C$, of the support of the conditional distribution of $L$, given $\supersam$, 
because $\centr{\supersam}{L} \le \log C$ a.s.
For all $i \in \{0,1\}$ and $j \in [n]$, write
$\supersam_{i,j} = (X_{i,j},Y_{i,j})$, and
consider the set of possible predictions on $\supersam$,
\[
P = \{ p \in \{0,1\}^{\{0,1\} \times [n]} : \exists h \in \hypothesisclass_n, \forall i \in \{0,1\}, j \in [n], p_{i,j} = h(X_{i,j}) \} .
\]
The set $P$ is precisely the set of possible labellings of the $2n$ inputs $(X_{i,j})$, which is bounded by the growth function of $\HS_n$ evaluated at $2n$ points.
By the Sauer--Shelah lemma, the cardinality of $P$ is thus bounded above by $6n^{\vcd_n}$.
Note that the support of the conditional distribution of $L$ given $\supersam$ is
\[
\{ c \in \{0,1\}^{\{0,1\} \times [n]} : \exists p \in P, \forall i \in \{0,1\}, \forall j \in [n], \text{ if and only if } c_{i,j} = \indic{p_{i,j} \neq Y_{i,j}}\}.
\]
Therefore, the cardinality of the support is no greater than that of $P$, hence $C \le 6n^{\vcd_n}$.

\section{Proof of \cref{mainthm}}
\label{pf:sample-complexity-contradiction}
We prove the claim by contradiction. Pick $f$ and $c \ge 0$.
Let $\HS$ be a concept class with finite VC dimension $\vcd$ as shown to exist by \cref{thm:sample-complexity-lowerbound}. 

Let $\TheAlg$ be a proper learning algorithm for $\HS$,
let $n \ge \vcd$, and assume, for the eventual purpose of obtaining a contradiction, 
that
$\SZB \le f(\vcd)$ for all $\Dist$ and,
for all $s \in \dataspace^{n}$, 
$\EE \EmpRisk{s}{\Alg(s)} \le c\, {\vcd}/{n}$
if there exists $h \in \hypothesisclass$ such that $\EmpRisk{s}{h}=0$. Pick a realizable distribution $\Dist$. 
It follows from the above assumption and \cref{eq:ege-interpolate} that
$ \EE [\Risk{\Dist}{\Alg(\SS)}]\leq 2 c\, {\vcd}/{n} + 3{f(\vcd)}/{n } = {(3f(\vcd)+2 c\,\vcd)}/{n}$
By Markov's inequality, $\Pr (\Risk{\Dist}{\Alg(\SS)} \geq \epsilon)\leq \frac{1}{n\epsilon}(3f(\vcd)+2c\,\vcd)$.
It follows that the sample complexity of proper learning $\HS$ satisfies
\[
\label{eq:sample-complexity-upperbound}
\samplecomplex{\epsilon}{\delta}\leq  \frac{1}{\epsilon \delta}(3f(\vcd)+2c\,\vcd).
\]
Now, fix $\delta \in (0,1/100)$ and fix a convergent sequence of $\epsilon_i \downto 0$. 
There exists $J$ such that, for all $i \geq J$,
\[\label{eq:whatever}
\frac{1}{\tilde{c} \delta}(3f(\vcd)+2c\,\vcd)  < \vcd\, \mathrm{Log} \frac{1}{\epsilon_i} +  \mathrm{Log}\frac{1}{\delta} ,
\] 
 for $\tilde{c}$ as in \cref{thm:sample-complexity-lowerbound}.
Combining \cref{eq:whatever} with \cref{eq:sample-complexity-upperbound},
$ \samplecomplex{\epsilon_i}{\delta}
     < \frac{\tilde{c}}{\epsilon_i}(d \, \mathrm{Log} \frac{1}{\epsilon_i} + \mathrm{Log}\frac{1}{\delta})$
for $i \geq J$.
Simultaneously, from \cref{thm:sample-complexity-lowerbound}, it follows that 
$\samplecomplex{\epsilon_i}{\delta}\geq  \frac{\tilde{c}}{\epsilon_i}(d\, \mathrm{Log} \frac{1}{\epsilon_i} + \mathrm{Log}\frac{1}{\delta})$, a contradiction.

\newcommand{\disag}[1]{\mathrm{DIS}(#1)} 

\section{Proof of \cref{thm:publish-versionspace} }
\label{pf:publish-versi}
Let $n \geq \starnum$. First we begin with two definitions: 
letting $\SS$ be the random element in $(\inspace\times \outspace)^n$ representing
our training sample,
the \defn{empirical teaching
  dimension} \citep{hanneke2007teaching,el2010foundations}, denoted $\teachdim_n$, is 
  size of the smallest subset of $S$ that produces the same version space, i.e.,
\[
  \teachdim_n=\min\{|S'|: S'\subseteq \SS,
  \versionspace{S'}=\versionspace{\SS}\}.
\]
An important fact about the empirical teaching dimension is that $\teachdim_n$ is
bounded by star number almost surely \citep{hanneke2016refined}. 
An \defn{empirical teaching set} is any subset of $\SS$ that achieves the minimum in the definition of the
 empirical teaching dimension. 
Consider a realizable training set $\SS\in (\inspace\times \outspace)^n$. Let $S'$ denote an empirical teaching set of $\SS$. Let $\hat{S} \subseteq \SS \setminus S'$. By the definition of the version space we have $\versionspace{\SS} \subseteq \versionspace{\SS \setminus \hat{S}} \subseteq \versionspace{S'}$. Also, by the definition of $S'$, we have $\versionspace{\SS}=\versionspace{S'}$; therefore, $\versionspace{\SS} = \versionspace{\SS \setminus \hat{S}}$. This argument shows that the version space is \emph{stable}, in the sense that removing any point that is not in $S'$ does not alter the version space. 
Let $\mathrm{ETS}(\SS)$ denote the empirical teaching set $\SS$. 

We also need the definition of the \emph{region of disagreement} denoted by $\disag{\versionspace{\SS}}=\{x \in \inspace \vert \exists h, h' \in \versionspace{\SS}~\text{such that}~ h(x)\neq h(x')\}$.

Following the same line of reasoning as in \cref{eq:lemma-lowerbound-use} we obtain $\minf{\versionspace{Z_U};U\vert Z} \leq  n \log 2 -  \sum_{i=1}^{n} \entr{U_i | U_{-i}, \versionspace{Z_U}, \supersam}$. Since the order of the training set does not change the version space we get $\sum_{i=1}^{n}\entr{U_i | U_{-i}, \versionspace{Z_U}, \supersam}=n \entr{U_1 | U_{-1}, \versionspace{Z_U}, \supersam}$.  Fix $i \in \range{n}$, and define $\Ur{i}{b}\triangleq (U_1,\dots,U_{i-1},b,U_{i+1},\dots,U_n)$ for $b \in \{0,1\}$. Using this notation we can define training sets  $\Sr{i}{b}= \supersam_{\Ur{i}{b}}$ for $b \in \{0,1\}$ and $i \in [n]$. Let $\mathcal{F}_1 = \sigma(U_{-1}, \versionspace{Z_U}, \supersam)$. Then, we have 
\[
\label{eq:etd-pf-condentropy}
\entr{U_1 | U_{-1}, \versionspace{Z_U}, \supersam}\geq \EE \big[ \centr{\cF_1 }{U_1} \indic{\versionspace{\Sr{1}{0}}=\versionspace{\Sr{1}{1}}}  \big].
\]
Using the same techniques as in the proof \cref{thm:stable-compression-cmi}, we can show that on the event $\indic{\versionspace{\Sr{1}{0}}=\versionspace{\Sr{1}{1}}}$, $\centr{\cF_1 }{U_1}=\log(2)$. Thus, $\entr{U_1 | U_{-1}, \versionspace{Z_U}, \supersam} \geq  \Pr(\versionspace{\Sr{1}{0}}=\versionspace{\Sr{1}{1}}) \log 2$.

For $i \in [n]$, define the training set $S_{-i} \triangleq \{Z_{U_1},\dots,Z_{U_{i-1}},Z_{U_{i+1}},\dots, Z_{U_{n}}\}$. Recall that $Z_{i,j}=(X_{i,j},Y_{i,j})$. We claim that 
\[
\versionspace{\Sr{1}{0}}\neq\versionspace{\Sr{1}{1}} \Rightarrow (X_{0,1} \in  \disag{\versionspace{S_{-1}}}) \vee (X_{1,1} \in  \disag{\versionspace{S_{-1}}}) \label{eq:claim-version-dis}.
\]
We prove this claim by  contraposition. Given that $(X_{0,1} \notin  \disag{\versionspace{S_{-1}}}) \wedge (X_{1,1} \notin  \disag{\versionspace{S_{-1}}})$, we have that the concepts in $\disag{\versionspace{S_{-1}}}$ agree for prediction of $X_{0,1}$ and $X_{1,1}$. As $\Dist$ is a realizable distribution, we conclude   $\versionspace{\Sr{1}{0}} =  \versionspace{S_{-1}} =  \versionspace{\Sr{1}{1}}$. 

In the next step, we provide an upper bound on $\Pr((X_{0,1} \in  \disag{\versionspace{S_{-1}}}) \vee (X_{1,1} \in  \disag{\versionspace{S_{-1}}}))$ as follows
\[
&\Pr((X_{0,1} \in  \disag{\versionspace{S_{-1}}}) \vee (X_{1,1} \in  \disag{\versionspace{S_{-1}}})) \nonumber \\
&\leq \Pr((X_{0,1} \in  \disag{\versionspace{S_{-1}}})) + \Pr((X_{1,1} \in  \disag{\versionspace{S_{-1}}})) \nonumber\\
&= 2 \Pr(X_{0,1} \in  \disag{\versionspace{S_{-1}}}). \label{eq:disag-eq-union-iid}
\]
Here, we have used the union bound and the points in $\supersam$ are \iid. Then, we can write
\[
\Pr(X_{0,1} \in  \disag{\versionspace{S_{-1}}}) &= \EE{ \sbra{ \indic{X_{0,1} \in  \disag{\versionspace{S_{-1}}}}}} \nonumber \\
                                                &= \EE{ \sbra{ \indic{X_{1} \in  \disag{\versionspace{\{Z_2,\dots,Z_n\}}}}}} .
\]
The last step follows from  $U$ and $\supersam$ are independent, and the points in $\supersam$ are \iid. Therefore, in the last step we consider the expectation over $(Z_1,\dots,Z_n)\dist \Dist^{\otimes n}$. Note that $Z_i = (X_i, Y_i)$. By the exchangeability of the points in $(Z_1,\dots,Z_n)$ we have
\[
\EE{ \sbra{ \indic{X_{1} \in  \disag{\versionspace{\{Z_2,\dots,Z_n\}}}}}} &= \frac{1}{n} \sum_{i=1}^{n}  \EE{ \sbra{ \indic{X_{i} \in  \disag{\versionspace{ \{Z_1,\dots,Z_n\}\setminus \{Z_i\} }}}}}.
\]
Then, we claim that given  $X_{i} \in  \disag{\versionspace{ \{Z_1,\dots,Z_n\}\setminus \{Z_i\} }}$ then $X_i \in  S'$ where $S'$ is \emph{any} teaching set of $\{Z_1,\dots,Z_n\}$. We can easily prove this claim by contradiction and the stability of the version space shown in the beginning of this section. Therefore, 
\[
 \frac{1}{n} \sum_{i=1}^{n}  \EE{ \sbra{ \indic{X_{i} \in  \disag{\versionspace{ \{Z_1,\dots,Z_n\}\setminus \{Z_i\} }}}}} &\leq  \frac{1}{n} \sum_{i=1}^{n}  \EE{ \sbra{ \indic{X_{i} \in  S' }}} \nonumber\\
&=\frac{1}{n} \EE{ \sbra{ \sum_{i=1}^{n}   \indic{X_{i} \in  S' }}} \nonumber\\
&\leq \frac{\starnum}{n} .\label{eq:final-starnumber-exch}
\]
Here, we have used the linearity of the expectation, and the last step follows from the fact that the cardinality of $S'$ is at most $\starnum$ almost surely. By \cref{eq:disag-eq-union-iid}-\cref{eq:final-starnumber-exch}, we have $\Pr((X_{0,1} \in  \disag{\versionspace{S_{-1}}}) \vee (X_{1,1} \in  \disag{\versionspace{S_{-1}}}))  \leq {2\starnum}/{n} $. Then, by \cref{eq:claim-version-dis} we have 
\[
  \entr{U_1 | U_{-1}, \versionspace{Z_U}, \supersam} &\geq  \Pr(\versionspace{\Sr{1}{0}}=\versionspace{\Sr{1}{1}}) \log 2 \nonumber\\
  &\geq \sbra{ 1-\Pr((X_{0,1} \in  \disag{\versionspace{S_{-1}}}) \vee (X_{1,1} \in  \disag{\versionspace{S_{-1}}}))} \log 2 \nonumber \\
  &\geq   ( 1 - \frac{2\starnum}{n})\log 2.
\]
Finally, combining this results and $ \minf{\versionspace{Z_U};U\vert Z} \leq  n \log 2 - n \entr{U_1 | U_{-1}, \versionspace{Z_U}, \supersam}$, we obtain $\minf{\versionspace{Z_U};U\vert Z} \leq 2\starnum\log 2   $ which was to be shown.

\section{Proof of \cref{thm:conv-publish-versionspace}}
\label{pf:conv-publish-versionspace}
We begin the proof by a lemma, which can be seen as the generalization of the well-known coupon collector's problem \citep[][Lem.~19]{bousquet2020proper}.
\begin{lemma}
\label{lemma:coupon-collector}
Let $M,m \in \Naturals$ and $1 \leq m \leq M$. Assume we take $k$ samples
$X_1,\dots,X_k$ uniformly at random with replacement from $[M]$. Then
$\Pr ( |[M]\setminus \{X_1,\dots,X_k\}| \geq m) \geq 1/2$
if $k \leq (M/2) \log \frac{M}{m}$.
\end{lemma}
Note that $|[M]\setminus \{X_1,\dots,X_k\}|$ is the number of unseen elements from $[M]$ after observing samples $X_1,\dots,X_k$.  

Consider a concept class $\HS$ over the input
space $\inspace$ whose star number is $\starnum$. Let $M = \min\{n,\starnum\}$.  By the definition of star
number, there exists $(x_1,\dots,x_M) \in \inspace^M$,  such that
$((x_1,y_1),\dots,(x_M,y_M))$ is realizable by $h^{\star}_0 \in \HS$, and every
neighbour of this sequence is also realizable by a classifier in $\HS$. Let
$h_{x_i} \in \HS $ be any classifier such that $\{j\in [M]\vert h_{x_i}(x_j)\neq y_j\}=\{i\}$.

For the case $M=1$, our lower bound is simply $\min\{n,\starnum\}-1=M-1=0$ which is trivial since the mutual information is non-negative. 
Therefore, in the rest of the proof we assume $M\geq 2$.

For the case  $M\geq 2$, consider the following distribution on the input space
\*[
\Dist_X(x_1)=1-\frac{M-1}{n} \text{ and }  \Dist_X(x_i)=\frac{1}{n},\ \text{ for $i \in\{2,\dots,M\}$.}
\]
Also let the target function (labelling function) be $h^{\star}_0$. Let $Z$, $U$, and $S$ be defined as usual, based on a sample from $\Dist_{X}$
labeled by $h^{\star}_0$.
Since $\Dist_{X}$ has zero measure on $\inspace \setminus \{x_1,\dots,x_M\}$,
we can, without any loss of generality, assume that $\HS =
\{h^\star_0,h_{x_1},\dots,h_{x_M}\}$,
as every other classifier is equivalent to one of these $\Dist_{X}$-almost
everywhere.
When we release the version space, we can agree in advance that elements stand
for their equivalence classes, which does not affect the information content.

Let $X=(X_{i,j})_{i\in \{0,1\},j \in [n]}$ denote the inputs observed in the
supersample $Z$. Define $X_U$ as the sequence of the inputs observed in $Z_U$.
Similarly, let  $X_{\bar{U}}$ denote the sequence of inputs observed in the ``ghost sample'' $Z_{\bar{U}}$,
where $\bar{U}$ denotes the sequence $U$ but with every entry flipped.
In the following, we write $x\in X_U$ to mean that there exists $i \in [n]$ such that $X_{U_i,i}=x$. 

We claim that
$\versionspace{Z_U}=\{h_{x_i}\vert i\in [M], x_i \notin  X_U \}
\cup\{h^\star_0\}$.\footnote{%
  Formally, this statement holds almost surely. We will
  skip such declarations for the remainder of the proof.}
To see this, note that, if $x_i \notin X_U$, then $h_{x_i}=h^\star_0$ on $S$ and so both
are ERMs. In the other direction, if $x_i \in X_U$, then $h_{x_i}$ makes at least
one mistake, and so is not an ERM.
Thus, if $\versionspace{Z_U}$ has a non-empty intersection with
$\{h_{X_{0,i}},h_{X_{1,i}}\}$, we can perfectly recover $U_i$ from $\versionspace{Z_U}$, since in each
column of the supersample $Z$, only one point is selected for the training set.
 We use this observation to lower bound the conditional entropy of $U$ given $Z$
 and the version space.

 To that end, let $J=\{j \in [n] \vert \versionspace{Z_U} \cap \{h_{X_{0,j}},h_{X_{1,j}}
 \}\neq \emptyset \}$ be the set that contains the index of every column $j$ for
 which we can perfectly recover $U_j$ from the version space.
 Note that we can represent $J$ in an equivalent form as follows. For two sequences $v,w
\in \{x_1,\dots,x_M\}^n$, define $K(v,w)=\{j \in [n]\vert w_j \neq v_i \
\text{for all} \ i\in [n] \}$. 
By the definition of the version space,  $x_i \notin X_U$ is equivalent to  $h_{x_i} \in \versionspace{Z_U}$. Let $j \in [n]$, then
\*[
j \in J \Leftrightarrow \versionspace{Z_U}\cap \{h_{X_{1-U_j,j}}\}=\{h_{X_{1-U_j,j}}\} \Leftrightarrow X_{1-U_j,j} \notin X_U \Leftrightarrow j \in K(X_U,X_{\bar{U}}).
\]
Therefore, $J = K(X_U,X_{\bar{U}})$.

By the definition of the mutual information, and the fact that $U$ and $Z$ are independent, we have
\[
\minf{\versionspace{Z_U};U\vert Z}= n\log (2) - \entr{U\vert Z,\versionspace{Z_U}} \label{eq:mi-versionspace}.
\]
Then
\begin{align}
    \entr{U\vert Z,\versionspace{Z_U}} &= \entr{U\vert Z,\versionspace{Z_U},J} \label{eq:determinsticfun-j}\\
                                    &=\entr{U_J, U_{J^c}\vert Z,\versionspace{Z_U},J} \\
                                    &=\entr{U_{J^c}\vert Z,\versionspace{Z_U},J} \label{eq:j-perfectrecover} \\
                                    &\leq \EE[n-|J|] \log(2)  \label{eq:lower-bound-versionspace-cardinality},
\end{align}
where
\cref{eq:determinsticfun-j} follows from $J$ being known from $Z$ and $\versionspace{Z_U}$;
\cref{eq:j-perfectrecover} follows from $U_J$ being known from $J$, $Z$, $\versionspace{Z_U}$;
and \cref{eq:lower-bound-versionspace-cardinality} follows from
the cardinality of the support of the distribution of $U_{J^c}$ being no more than $2^{(n-|J|)}$.
Therefore, by \cref{eq:mi-versionspace} and \cref{eq:lower-bound-versionspace-cardinality}, we have 
\[
\label{eq:minf-version-j}
\minf{\versionspace{Z_U};U\vert Z} \geq \EE[J] \log 2.
\]

In the next step of the proof, we lower bound $\EE[|J|]$. 
Because $U$ and $Z$ are independent and $Z$ is an i.i.d.\ array, $X_U$ and
$X_{\bar{U}}$ are independent and identically distributed sequences of i.i.d.\
elements with common distribution $\Dist_\inspace$.
Thus, $\EE[|J|] =\EE \big[ |K(X,\bar{X})| \big]$,
where $X$ and $\bar{X}$ are i.i.d.\ copies of $X_U$ (equivalently,
$X_{\bar{U}}$).

Let $\hat{M}$ be the number of elements of $\{x_2,\dots,x_M\}$ \emph{not} appearing in $X$,
i.e., $\hat{M}=|\{i \in \{2,\dots,M\}\vert x_i \neq X_j \ \text{for all} \ j \in [n]\} |$.
Conditional on $X$,
if $x_1 \in X$, then $|K(X,\bar{X})|$ is a Binomial random variable with $n$ trials,
each succeeding with probability $\frac{\hat{M}}{n}$.
If $x_1 \notin X$, then $ |K(X,\bar{X})|$ is a Binomial random variable with $n$
trials, each succeeding with probability
$\frac{\hat{M}}{n} + 1-\frac{M-1}{n} \ge \frac{\hat {M}}{n}$.
Therefore, 
\[
\EE[|J|] &=\EE \big[ |K(X,\bar{X})| \big]\\
        & = \EE\big[ \cEE{X}{[|K(X,\bar{X})|]}\big] \label{eq:pf-versionspace-towerrule}\\
        & \geq \EE[\hat{M}] \label{eq:pf-versionspace-bionomial}.
\]
Here,  \cref{eq:pf-versionspace-towerrule} follows from the tower rule and
\cref{eq:pf-versionspace-bionomial} follows from the fact that the mean of the
binomial distribution with $n$ trials, each succeeding with probability $p$, is
$np$.

Fix $\beta=\exp(-3)$ and $m=\beta (M-1)$. Let $\hat{N}$ denote the number of samples in $X$ falling in $\{x_2,\dots,x_M\}$. Then, using the tower rule and Markov's inequality we have
\[
\EE[\hat{M}]&= \EE \big[\cEE{\hat{N}}{[\hat{M}]} \big]\\
&\geq  \EE\big[ m\cPr{\hat{N}}{\hat{M}\geq m}\big] \label{eq:markov-inequality-condition} \\
&\geq  \EE\big[ m\cPr{\hat{N}}{\hat{M}\geq m} \indic{\hat{N}\leq ({(M-1)}/{2})\log\beta^{-1}}\big]
\]
Note that, conditional on $\hat{N}$, the $\hat{N}$ samples falling in $\{x_2,\dots,x_M\}$ are conditionally independent, with conditional distribution
uniform on this set. Using this fact, from  \cref{lemma:coupon-collector} we obtain
\[
\EE\big[ m\cPr{\hat{N}}{\hat{M}\geq m} \indic{\hat{N}\leq ({(M-1)}/{2})\log\beta^{-1}}\big] \geq  \frac{1}{2} m \Pr(\hat{N}\leq ({(M-1)}/{2})\log\beta^{-1})\label{eq:coupon-collector}.
\]
We have $\EE[\hat{N}]=M-1$ and, by Markov's inequality, $\Pr(\hat{N}\leq \frac{M-1}{2}\log\beta^{-1}) \geq 1-2/{\log(\beta^{-1})}=1/3$. By combining \cref{eq:minf-version-j}, \cref{eq:pf-versionspace-bionomial}, \cref{eq:coupon-collector}, and $\Pr(\hat{N}\leq \frac{M-1}{2}\log\beta^{-1}) \geq 1/3$ we get
\*[
\minf{\versionspace{Z_U};U\vert Z} & \geq  \frac{\beta\log 2}{6} (M-1)  \\
&=\Omega(\min\{n,\starnum\}-1),
\]
which was to be shown.
\section{Proof of {\cref{thm:starnumber-log}}}
\label{pf:starnumber-log}
By the well-ordering theorem
\citep{zermelo1908untersuchungen}, there exists a binary
relation $\ll$ on $\HS$ that is transitive, total, antisymmetric, and
well-ordered. In particular, the well-ordered property implies that every nonempty subset $\HS$ has the least element. The proposed learning algorithm is given by 
\*[
\Alg(\SS) = \text{LE}(\versionspace{\SS}),
\]
where $\text{LE}$ of a nonempty set denotes its least element with respect to
$\ll$.\footnote{In some cases this classifier may not be measurable. We will
  assume it is. To avoid measure-theoretic issues, one may assume $\inspace$ is
  countably infinite or finite, or design a well-ordering by hand to guarantee
  measurability if possible.} 
Note that $\Alg$ is \emph{deterministic}, \emph{consistent}, and \emph{permutation-invariant},  i.e.,  the order of the points in $\SS$ does not impact the output. \\
Let $W=\Alg(Z_U)$. By the definition of the mutual information and \cref{lemma:chainrule-lowerbound}, we obtain
\[
\SZB &= H(U|\supersam) - H(U|W,\supersam) \nonumber \\ 
&\leq n \log(2) - \sum_{i=1}^{n} H(U_i | U_{-i},\supersam,W) \label{eq:erm-star-cmi}.
\]
The last step follows since $U$ is independent of $\supersam$ and the independence of indices of $U$ and \cref{lemma:chainrule-lowerbound}. By the permutation-invariance of the algorithm, we have $\sum_{i=1}^{n} H(U_i | U_{-i},\supersam,W)=n H(U_1 | U_{-1},\supersam,W)$. Fix $i \in \range{n}$, and define $\Ur{i}{j}\triangleq (U_1,\dots,U_{i-1},j,U_{i+1},\dots,U_n)$ for $j \in \{0,1\}$. Using this notation we can define  training set  $\Sr{i}{j}= \supersam_{\Ur{i}{j}}$ for $j \in \{0,1\}$ and $U_{-i}\in \{0,1\}^{n-1}$.  
 Let $\mathcal{F}_1 = \sigma(W,U_{-1},\supersam)$.  Define the $\mathcal{F}_1$-measurable event $\mathcal{E}=\{\Alg(\Sr{1}{0})=\Alg(\Sr{1}{1})\}. $ Then, we can write $H(U_1 | U_{-1},Z,W)=\EE[ \centr{\mathcal{F}_1 }{U_1}(\indic{\mathcal{E}} + \indic{\mathcal{E}^c})]$.
We claim the following facts:
\begin{enumerate}
\item On the event $\mathcal{E}^c$,   $\centr{\mathcal{F}_1 }{U_1}=0$ \as.
\item On the event $\mathcal{E}$,   $\centr{\mathcal{F}_1 }{U_1}=\log 2$ \as.
\end{enumerate}
The reason is that since the algorithm is deterministic, on the event $\mathcal{E}^c$ we can perfectly recover $U_1$ since $W$ is either $\Alg(\Sr{1}{0})$ or $\Alg(\Sr{1}{1})$ which shows that $\centr{\mathcal{F}_1 }{U_1}=0$. Then, on the event $\mathcal{E}$, using the Bayes rule we can show that $\centr{\mathcal{F}_1 }{U_1}=\log 2$.
Therefore, we have $\entr{U_1 \vert U_{-1},\supersam,W}= \EE[\indic{\mathcal{E}}] \log(2)$. We can further lower bound 
\[
\label{eq:cond-entr-erm}
\entr{U_1 \vert U_{-1},\supersam,W}\geq \EE\big[\indic{\mathcal{E}}\indic{\loss(\Alg(\Sr{1}{0}),\supersam_{1,1})=0 \wedge \loss(\Alg(\Sr{1}{1}),\supersam_{0,1})=0}\big]  \log(2).
\]
Next, we claim that on the event  $\{\loss(\Alg(\Sr{1}{0}),Z_{1,1})=0\}\wedge \{\loss(\Alg(\Sr{1}{1}),Z_{0,1})=0\}$, we have $\Alg(\Sr{1}{0})=\Alg(\Sr{1}{1})$ \as .  This claim can be proved by contradiction. Assume $\Alg(\Sr{1}{0})\neq \Alg(\Sr{1}{1})$. Consider the version space $\versionspace{\Sr{1}{0} \cup \Sr{1}{1}}$.  By the assumptions   $\loss(\Alg(\Sr{1}{0}),Z_{1,1})=0$ and $\loss(\Alg(\Sr{1}{1}),Z_{0,1})=0$, we have $\Alg(\Sr{1}{0}) \in \versionspace{\Sr{1}{0} \cup \Sr{1}{1}}$ and $\Alg(\Sr{1}{1}) \in \versionspace{\Sr{1}{0} \cup \Sr{1}{1}}$.   Also, it is immediate to see that  $\versionspace{\Sr{1}{0} \cup \Sr{1}{1}} \subseteq \versionspace{\Sr{1}{0} } $ and $\versionspace{\Sr{1}{0} \cup \Sr{1}{1}} \subseteq \versionspace{\Sr{1}{1} } $. Therefore, we have $\Alg(\Sr{1}{0}) \in  \versionspace{\Sr{1}{1}}$ and $\Alg(\Sr{1}{1}) \in \versionspace{\Sr{1}{0}}$. By the definition of the algorithm, we choose the least hypothesis from the version space. Thus,  $\Alg(\Sr{1}{0}) \ll \Alg(\Sr{1}{1})$ since $\Alg(\Sr{1}{1}) \in  \versionspace{\Sr{1}{0}}$. Similarly, we can show $\Alg(\Sr{1}{1}) \ll \Alg(\Sr{1}{0})$. Therefore, considering $\Alg(\Sr{1}{1}) \ll \Alg(\Sr{1}{0})$ and $\Alg(\Sr{1}{0}) \ll \Alg(\Sr{1}{1})$, we conclude that the assumption $\Alg(\Sr{1}{0})\neq \Alg(\Sr{1}{1})$ is false, a contradiction.

Having proved that $\Alg(\Sr{1}{0})= \Alg(\Sr{1}{1})$ on the event $\{\loss(\Alg(\Sr{1}{0}),Z_{1,1})=0\}\wedge \{\loss(\Alg(\Sr{1}{1}),Z_{0,1})=0\}$, we can further simplify \cref{eq:cond-entr-erm} as
\begin{align*}
H(U_1 | U_{-1},\supersam,W) &\geq \EE \left[ \indic{\loss(\Alg(\Sr{1}{1}),Z_{0,1})=0}  \indic{\loss(\Alg(\Sr{1}{0}),Z_{1,1})=0}\right]  \log 2\\
&\geq  (1- \EE [\indic{\loss(\Alg(\Sr{1}{1}),Z_{0,1})=1} ] - \EE [\indic{\loss(\Alg(\Sr{1}{1}),Z_{0,1})=1} ])\log 2\\
&=  (1-2 \EE[\Risk{\Dist}{\Alg(\SS)}]\log 2.
\end{align*}
The last step follows since the elements of $\Sr{1}{0}$ and $\Sr{1}{1}$ are \iid.  By plugging this lower bound into \cref{eq:erm-star-cmi}, we obtain
\[
\label{eq:cmi-simple-erm}
\SZB  \leq 2n\log(2) \  \EE[\Risk{\Dist}{\Alg(\SS)}].
\]
Then, we use the the result from Corollary 12 of  \citep{hanneke2016refined} which states the following bound holds uniformly for the expected risk of \emph{all} consistent and proper learners for learning a concept class with VC dimension $\vcd$ and star number $\starnum$:
\*[
\EE[\Risk{\Dist}{\Alg(\SS)}] =O(\frac{\vcd}{n} \log(\frac{\min\{\starnum,n\}}{\vcd})).
\]
Finally, the stated result follows by combining this bound with  \cref{eq:cmi-simple-erm}.

\section{Proof of \cref{thm:ecmi-consistent}}
\label{pf:ecmi-every-consistent}
Let $L = (L_1,\dots,L_n)$ where $L_i \in \{0,1\}^2$ denotes the vector at Column $i$ of the loss array $L$. 
By the definition of mutual information, 
we have $\minf{L;U\vert \supersam} \leq \entr{L\vert \supersam}$. 
Since conditioning decreases the entropy, 
we have $\entr{L\vert \supersam} \leq \entr{L}$. 
Finally, by the chain rule, $\entr{L} \leq  \sum_{i=1}^{n} \entr{L_i}$.

Note that $L_i$  takes values in $\{[1,0]^\intercal,[0,0]^\intercal,[0,1]^\intercal\}$, as it is assumed that $\Alg$ is consistent and, by construction, at each column, one of the points is selected for the training set. 
If $U_i=0$, then $Z_{0,i}$ is in the training set and so, because the algorithm is consistent, $L_{0,i}=0$ a.s. Therefore,
\*[
\Pr[L_{i}=[1,0]^\intercal] &=\EE\big[\cPr{U}{L_{i}=[1,0]^\intercal}\big] \\
&=\EE\big[ \indic{U_i=0} \cPr{U}{L_{i}=[1,0]^\intercal}\big] + \EE\big[ \indic{U_i=1} \cPr{U}{L_{i}=[1,0]^\intercal}\big] \\
&= \EE\big[ \indic{U_i=1} \cPr{U}{L_{i}=[1,0]^\intercal}\big] \label{eq:consistency}
\]
Conditioning on event $U_i=1$,
we have $L_{i}=[1,0]^\intercal = \loss(\Alg(Z_U),Z_{0,i})=1$
and therefore
\*[
\cPr{U}{L_{i}=[1,0]^\intercal}
= \cPr{U}{\loss(\Alg(Z_U),Z_{0,i})=1} 
= \Risk{\Dist}{\Alg}, 
\]
where the final equality follows from the definition of the expected risk and the fact that $Z_{0,i}$ is not in $Z_U$.
Therefore, $ \Pr[L_{i}=[1,0]^\intercal] = \EE[\indic{U_i=1}] \Risk{\Dist}{\Alg} = (1/2) \Risk{\Dist}{\Alg}$.
The same idea establishes that $\Pr[L_{i}=[0,1]^\intercal]=(1/2) \Risk{\Dist}{\Alg}$.

Therefore, by the definition of the entropy
\*[
\entr{L_i} &= -(1-\Risk{\Dist}{\Alg}) \log(1- \Risk{\Dist}{\Alg})  -  \Risk{\Dist}{\Alg} \log ( \frac{\Risk{\Dist}{\Alg}}{2})\\
& = \binaryentr{\Risk{\Dist}{\Alg}} + \Risk{\Dist}{\Alg} \log(2).
\]
The stated results follows from $\eSZB\leq \sum_{i=1}^{n}\entr{L_i}$.

\section{Proof of \cref{thm:singelton}}
\label{pf:loo}
We begin by presenting a theorem that will be used later to prove \cref{thm:singelton}. This theorem, which might be of independent interest, shows the \emph{average leave-one-error over supersample $Z$} can be used to upper bound $\eSZB$.

Lets begin by introducing some notations. For $n\in \Naturals$, let $[2n]_{n+1}$ denote the set of all size-$n+1$ subsets of $[2n]$ and let $\binaryentr{\cdot}$ denote the binary entropy function.
\begin{theorem}
\label{thm:one-inclusion-general}
 Let $\TheAlg$ denote a consistent and symmetric algorithm. Let $P_e: \dataspace^n \times \dataspace \to [0,1]$. For $s= ((x_1,y_1),\dots,(x_n,y_n)) \in \dataspace^{n}$ and $(x,y)\in \dataspace$, $P_e(s;(x,y))$ denotes the probability of error $\Alg$ for predicting the label $x$ where the randomness is over the internal randomness in $\TheAlg$.  
 Then, for every distribution $\Dist$, we have 
\begin{equation}
\label{eq:ecmi-oneinclusion}
\eSZB \leq n\EE \Big[ \binaryentr{\kappa(\supersamflat)} +\kappa(\supersamflat) \log 2 - \EE_{J \dist \unif{[2n]_{n+1}}} \frac{1}{n+1} \sum_{j \in J} \binaryentr{P_e(\supersamflat_{J-\{j\}};\supersamflat_j)} \Big]
\end{equation}
where 
$\supersamflat=(Z_1,\dots,Z_{2n})\dist \Dist^{\otimes (2n)}$ and 
$\kappa(\supersamflat) = \EE_{J \dist \unif{[2n]_{n+1}}} \frac{1}{n+1} \sum_{j\in J} P_e(\supersamflat_{J - \{j\}};\supersamflat_j)$ 
which takes values in $[0,1]$ almost surely.
Note, in \cref{eq:ecmi-oneinclusion}, the outer expectation is over $\supersamflat$.
\end{theorem}
The proof of \cref{thm:one-inclusion-general} is deferred to \cref{app:subsec:pf-auxilarythm}.
\subsection{Proof of \cref{thm:singelton}}
\label{app:subsec-pointfun}
First of all, note that given \emph{distinct} points $(x_1,\dots,x_{n+1}) \in \Reals^{n+1}$, the vertex set of the one-inclusion graph is given by
\*[
\{(h(x_1),\dots,h(x_{n+1})) \vert h \in \HS\} = \{(a_1,\dots,a_{n+1}) \vert a_i \in \{0,1\} \ \text{for all} \  i \in [n+1], \sum_{i \in [n+1]}a_i\leq 1\}.
\] 
As an example \cref{fig:one-inclusion-point} illustrates the structure of the one-inclusion graph for a sequence of \emph{distinct} points.
\begin{figure}[t]
\centering
\includegraphics[width=0.55\textwidth]{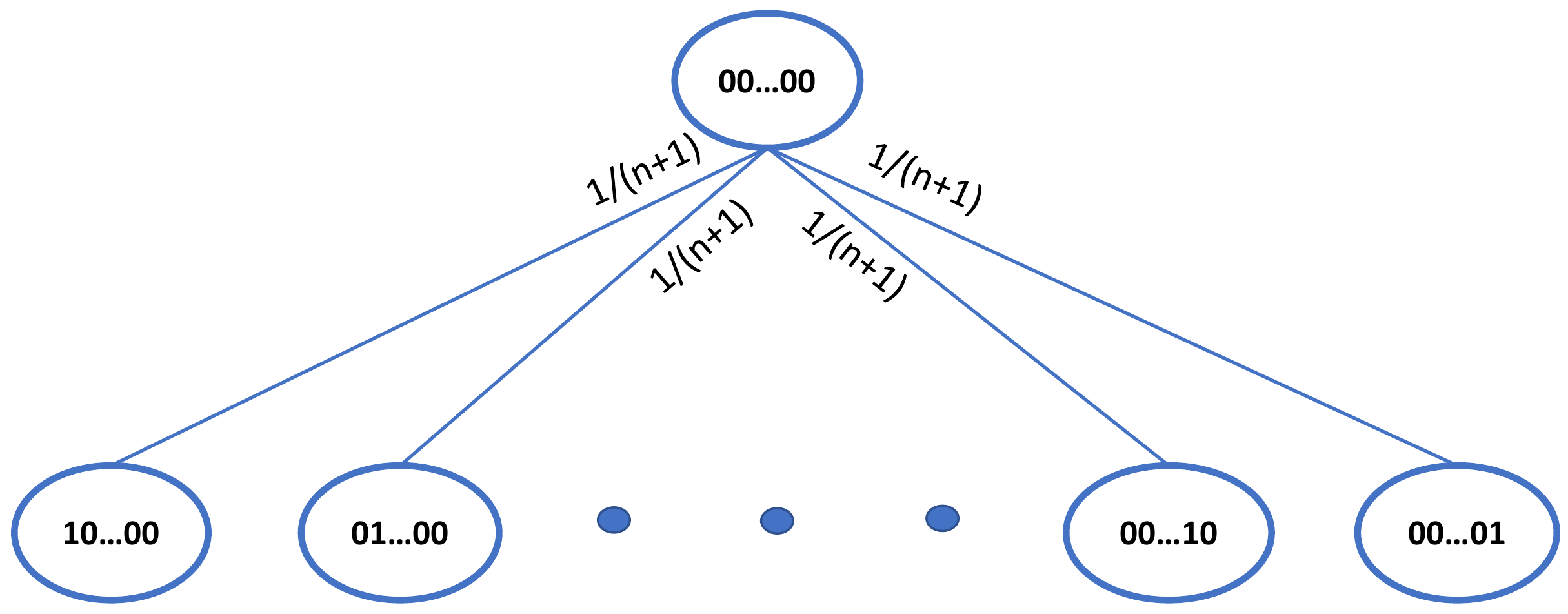}
\caption{One-inclusion graph of point functions for a set of distinct points.}
\label{fig:one-inclusion-point}
\end{figure}
The prediction rule is as follows. Given a training set $((x_1,y_1),\dots,(x_n,y_n))$ and test example $x$, we have two cases. If  $y_i=0$ for all $i \in [n]$, then the label of $x$ is predicted to be $0$ with probability $\frac{n}{n+1}$, and it is predicted to be one with probability $\frac{1}{n+1}$. The second case is that there exists a point with label $1$. Denote its index by $i^\star \in [n]$. In this case, the target function is known to be $\indic{x=x_{i^\star}}$.

To upper bound eCMI for this prediction rule, we consider the upper bound provided in \cref{thm:one-inclusion-general}. For a fixed supersample $\supersamflat=((x_1,y_1),\dots,(x_{2n},y_{2n}))$, we will consider two cases:\\
\underline{Case I: For all $i \in [2n]$ we have $y_i=0$.} Let $J \in [2n]_{n+1}$ and
let $\mathrm{deg}(\supersamflat_J)$ denote the degree of the vertex $(0,0,\dots,0,0)$ in the one-inclusion graph $\supersamflat_J$. Note that if there is no repetition in $\supersamflat_J$ then, the degree is $n+1$. If repetition exists, it is less than $n+1$. Then, according to the prediction rule we have
\begin{equation}
\label{eq:qz-zero}
\begin{aligned}
\kappa(\supersamflat) &= \EE_{J \dist \unif{[2n]_{n+1}}} \frac{1}{n+1} \sum_{j\in J} P_e(\supersamflat_{J-\{j\}},\supersamflat_j)\\
&= \frac{1}{n+1} \EE_{J \dist \unif{[2n]_{n+1}}} \frac{\mathrm{deg}(\supersamflat_J)}{n+1}.
\end{aligned} 
\end{equation}
Also, 
\begin{equation}
\label{eq:entr-zero}
\EE_{J \dist \unif{[2n]_{n+1}}} \frac{1}{n+1} \sum_{j \in J} \binaryentr{P_e(\supersamflat_{J-\{j\}},\supersamflat_j)} = \binaryentr{\frac{1}{n+1}} \EE_{J \dist \unif{[2n]_{n+1}}} \frac{\mathrm{deg}(\supersamflat_J)}{n+1} .
\end{equation}
let $\alpha = \EE_{J \dist \unif{[2n]_{n+1}}} \frac{\mathrm{deg}(\supersamflat_J)}{n+1}$. If $\alpha=0$, then $\eSZB = 0=O(1)$. Therefore, assume $\alpha > 0$, then we can use \cref{eq:ecmi-oneinclusion}, \cref{eq:qz-zero}, and \cref{eq:entr-zero} to obtain
\begin{align}
\eSZB &\leq  n \EE \big[ \binaryentr{\frac{\alpha}{n+1}} + \frac{\alpha \log2 }{n+1} - \alpha \binaryentr{\frac{1}{n+1}} \big] \label{eq:first-line-zero}\\
&\leq  n \EE\big[-\frac{\alpha}{n+1} \log \frac{\alpha}{n+1} - (1-\frac{\alpha}{n+1}) \log(1-\frac{\alpha}{n+1}) + \frac{\alpha \log2 }{n+1} + \nonumber \\
& \quad \quad  \frac{\alpha}{n+1}\log(\frac{1}{n+1}) + \alpha (1-\frac{1}{n+1})\log(1-\frac{1}{n+1})\big]\\
&\leq n \EE\big[ -\frac{\alpha}{n+1} \log(\alpha) + \frac{2\alpha  \log 2 }{n+1} \big] \label{eq:lots-simple}\\
&\leq n \big[\frac{2 \log 2}{n+1} + \frac{\exp(-1)}{n+1} \big] \label{eq:maxxlogx}\\
&=O(1) \label{eq:last-line-zero}.
\end{align}
\cref{eq:lots-simple} is obtained by removing the negative terms and the inequality $-(1-x)\log(1-x)\leq x \log 2$ for $x\in [0,1]$. \cref{eq:maxxlogx} follows from $\max_{x\in[0,1]}-x\log(x)=\exp(-1)$ and $\alpha\leq 1$.\\
\underline{Case II: There exists $i \in [2n]$ such that $y_i=1$.} Let $m = |\{i \in [2n] \vert y_i=1\}|$, $\beta = \frac{ \binom{2n-m}{n+1} }{ \binom{2n}{n+1} }$, 
and 
$\alpha = \EE_{J}[\frac{\mathrm{deg}(\supersam_J)}{n+1}\vert \text{no point with label 1 in $\supersamflat_J$}]$. 
Then, we have
\[
\label{eq:qz-one}
\kappa(\supersamflat) &= \EE_{J \dist \unif{[2n]_{n+1}}} \frac{1}{n+1} \sum_{j\in J} P_e(\supersamflat_{J-\{j\}},\supersamflat_j) \nonumber \\
& = \beta \EE_{J}[ \frac{1}{n+1} \sum_{j\in J} P_e(\supersamflat_{J-\{j\}},\supersamflat_j)\vert \text{no point with label 1 in $\supersamflat_J$}]  \nonumber \\
& + (1-\beta) \EE_{J}[ \frac{1}{n+1} \sum_{j\in J} P_e(\supersamflat_{J-\{j\}},\supersamflat_j)\vert \text{there exists points with label 1 in $\supersamflat_J$}] \nonumber \\ 
&\leq  \alpha \beta \frac{1}{n+1} + (1-\beta) \frac{1}{(n+1)^2} ,
\] 
where in the last step we have used the fact that when in $\supersamflat_J$ there is a point with label $1$, then the leave-one-out error is at most $\frac{1}{(n+1)^2}$ and  $\beta = \Pr(\text{no point with label 1 in $\supersamflat_J$})=\frac{ \binom{2n-m}{n+1} }{ \binom{2n}{n+1} }$.

The binary entropy function is concave and increasing in the interval $[0,1/2]$. The concavity of the binary entropy function implies that $\binaryentr{v_2}\leq \binaryentr{v_1}'(v_2-v_1)+\binaryentr{v_1}$ for all $v_1$ and $v_2$ in $(0,1)$. Then,
\[
\binaryentr{\kappa(\supersamflat)} & \leq  \binaryentr{\alpha \beta \frac{1}{n+1} + (1-\beta) \frac{1}{(n+1)^2} } \nonumber \\
&\leq \binaryentr{ \frac{\alpha \beta}{n+1}} + \frac{1-\beta}{(n+1)^2} \log(\frac{n+1-\alpha \beta}{\alpha \beta}) \label{eq:term1-qz-one}.
\]
Then, by considering the summation only over $\supersamflat_J$ without any point with label $1$, we obtain
\[
\label{eq:term2-qz-one}
\EE_{J \dist \unif{[2n]_{n+1}}} \frac{1}{n+1} \sum_{j \in J} \binaryentr{P_e(\supersamflat_{J-\{j\}},\supersamflat_j)} \geq \alpha \beta \binaryentr{\frac{1}{n+1}}.
\]

Then, we can substitute \cref{eq:term1-qz-one} and \cref{eq:term2-qz-one} into \cref{eq:ecmi-oneinclusion} to obtain $\eSZB = O(1)$ following the same line of reasoning as in \cref{eq:first-line-zero}--\cref{eq:last-line-zero}.

\subsection{Proof of \cref{thm:one-inclusion-general} }
\label{app:subsec:pf-auxilarythm}
We begin with introducing some notations. Let $\Gamma_{2n}$ be the set of all bijective mappings from $[2n]$ to $\{0,1\} \times [n]$. Let $\sigma \in \Gamma_{2n}$ and $x=(x_1,\dots,x_{2n})\in \inspace^{2n}$ be a vector of length $2n$. Then, $x^\sigma$ denotes a matrix of size $2 \times n$ where $x^\sigma_{i,j}=x_{\sigma^{-1}(i,j)}$ for $i \in \{0,1\}$ and $j \in [n]$. Also, for every $m, k \in \Naturals$ and $1\leq k \leq m$, let $[m]_k$ denote the set of all subsets of size $k$ of $[m]$.

For every $n \in \Naturals$ let $\pi \dist \unif{\Gamma_{2n}}$,  $\supersamflat=(Z_1,\dots,Z_{2n})\dist \Dist^{\otimes (2n)}$, and  $U = (U_1,\dots,U_n)\dist (\bernoulli(\{0,1\})^{\otimes n}$ where $\pi$, $\supersamflat$ and $U$ are mutually independent. Let $\SS=(\supersamflat^\pi_{U_j,j})_{j=1}^{n}$ and $\Alg(\SS)$ be a learning algorithm.  Let $L\in \{0,1\}^{2\times n}$ be a matrix with entries $L_{i,j}=\loss(\Alg(\SS),\supersamflat^{\pi}_{i,j})$ for $i\in \{0,1\}$ and $j\in \range{n}$. Then, using these random variables, we can define $\minf{L;U\vert \supersamflat,\pi}$.  
\begin{lemma}
$\eSZB = I(L;U|\supersamflat,\pi)$
\end{lemma}
\begin{proof}
$\supersamflat^\pi$ is a $\sigma(\supersamflat,\pi)$-measurable random variable therefore we have $I(L;U \vert \supersamflat,\pi)=I(L;U\vert\supersamflat^\pi,\supersamflat,\pi)$. Then, conditioned on  $\supersamflat^\pi$, $L$ and $U$ are independent from $\supersamflat$ and $\pi$. Finally, note that the samples in $\supersamflat$ are \iid, therefore we have $I(L;U\vert\supersamflat^\pi)=\eSZB$.
\end{proof}

\begin{lemma}
\label{lem:fun-sym}
Let $\pi$, $U$, and $\supersamflat$ be as defined in the beginning of this section. Let $n \in \Naturals$ and $f: \mathcal{Z}^{n}\times \mathcal{Z} \to \Reals $ be a real-valued function where $f$ is permutation-invariant with respect to its first input. Then
\[
\label{eq:permute_fun}
\cEE{\supersamflat}{\big[f((\supersamflat^{\pi}_{U_1,1},\dots,\supersamflat^{\pi}_{U_n,n});\supersamflat^{\pi}_{\bar{U_1},1})\big]} = \EE_{J \dist \unif{[2n]_{n+1}}} \frac{1}{n+1} \sum_{j\in J} f(\supersamflat_{J-\{j\}};\supersamflat_j),
\]
where $\supersamflat_J = (\supersamflat_{J_1},\dots,\supersamflat_{J_{n+1}}) $ and $\bar{U_i}=1-U_i$.
\end{lemma}
\begin{proof}
write
{\small
\begin{align}
    &\cEE{\supersamflat}{\big[f(\{ \supersamflat^{\pi}_{U_1,1},\dots,\supersamflat^{\pi}_{U_n,n}\};\supersamflat^{\pi}_{\bar{U_1},1})\big]} = \nonumber\\
& \sum_{ (i_1,\dots,i_{n+1})\in [2n]_{n+1}} f((\supersamflat_{i_1},\dots,\supersamflat_{i_{n}});\supersamflat_{i_{n+1}}) \cPr{\supersamflat}{(\{\supersamflat^{\pi}_{U_1,1},\dots,\supersamflat^{\pi}_{U_n,n}\},\supersamflat^{\pi}_{\bar{U_1},1})= (\{\supersamflat_{i_1},\dots,\supersamflat_{i_{n}}\}, \supersamflat_{i_{n+1}})}\nonumber.
\end{align}
}
Then, 
\*[
\cPr{\supersamflat}{(\{\supersamflat^{\pi}_{U_1,1},\dots,\supersamflat^{\pi}_{U_n,n}\},\supersamflat^{\pi}_{\bar{U_1},1})= (\{\supersamflat_{i_1},\dots,\supersamflat_{i_{n}}\}, \supersamflat_{i_{n+1}})}&=  \frac{n! (n-1)!}{(2n)!}\\
&  = \frac{1}{ \binom{2n}{n+1} (n+1)}
\]
where the last line follows since $U$ and $\pi$ are independent. It is easy to verify that is exactly equal to the RHS of \cref{eq:permute_fun}.
\end{proof}
Let $L = (L_1,\dots,L_n)$ where $L_i \in \{0,1\}^2$ denotes the vector at Column $i$ of the loss vector $L$. By the definition and the chain rule for mutual information  we have
\[
\minf{L;U\vert \supersamflat,\pi} &= \sum_{i=1}^{n} \minf{L_i;U\vert L_1,\dots,L_{i-1}, \supersamflat,\pi} \\
& = \sum_{i=1}^{n} \entr{L_i \vert  L_1,\dots,L_{i-1},\supersamflat,\pi} - \entr{L_i \vert L_1,\dots,L_{i-1},\supersamflat,\pi,U} \\
& \leq   \sum_{i=1}^{n} \entr{L_i \vert \supersamflat,\pi} - \entr{L_i \vert \supersamflat,\pi,U} \label{eq:upper-bounding-entr}\\
&= n(\entr{L_1 \vert \supersamflat,\pi} - \entr{L_1 \vert \supersamflat,\pi,U}). \label{eq:chainrule-ecmi}
\]
Here, \cref{eq:upper-bounding-entr} due to the Markov chain $L_i-(U,\supersamflat,\pi)- L_{-i}$ and removing conditions increases the entropy. Then, the last step follows since the algorithm is permutation-invariant. Note that if $\TheAlg$ is a deterministic algorithm the second term on the RHS of \cref{eq:chainrule-ecmi} is zero.

Next we provide an upper bound for $\entr{L_1 \vert \supersamflat,\pi}$. Note that $L_1$ can take values in $\{[1,0]^\intercal,[0,0]^\intercal,[0,1]^\intercal\}$ as it is assumed that $\Alg$ is consistent and, by construction, at each column of $\supersamflat^{\pi}$ one of the points is selected for the training set. 
Due to the symmetry imposed by $\pi$ and $U$ we have with probability one
\[
\label{eq:sym-error}
 \cPr{\supersamflat}{L_{1} = [1,0]^\intercal}= \cPr{\supersamflat}{L_{1} = [0,1]^\intercal}.
\]
Then, note that there exists a function $\kappa$,  depending only on $\Alg$,  such that 
\[
\label{eq:kappa-error}
\kappa(\supersamflat) = \cPr{\supersamflat}{L_{\bar{U_1},1} = 1}.
\]
We claim that the common value in \cref{eq:sym-error} is given by $\frac{\kappa(\supersamflat)}{2}  $. This claim can be easily proved by taking the expectation with respect to $U_1$ in \cref{eq:sym-error} and \cref{eq:kappa-error}.

Then, we will show that using \cref{lem:fun-sym}, $\kappa(\supersamflat)$ is given by the average leave-one-out error of $\TheAlg$ over $\supersamflat$. Given a training set $S=((x_1,y_1),\dots,(x_{n},y_{n})) \in \dataspace^n$ and a test point $z=(x,y)\in \dataspace$, let $P_e(S;z)$ denote the probability that $\Alg$ makes a mistake in predicting the label $x$. Using this notation, we have
\[
\label{eq:kappa-simplified}
\kappa(\supersamflat) &=  \cEE{\supersamflat}{\big[P_e((\supersamflat^{\pi}_{U_1,1},\dots,\supersamflat^{\pi}_{U_n,n});\supersamflat^{\pi}_{\bar{U_1},1})\big]} \\
                       &=\EE_{J \dist \unif{[2n]_{n+1}}} \frac{1}{n+1} \sum_{j\in J} P_e(\supersamflat_{J - \{j\}};\supersamflat_j), 
\]
where in the last step we have used \cref{lem:fun-sym}.
Then, by the fact that conditioning reduces the entropy we have $\entr{L_1\vert \pi,\supersamflat}\leq \entr{L_1\vert\supersamflat}$. 
As shown above, 
\begin{equation}
\label{eq:sym-prob-entropy}
\begin{aligned}
 \cPr{\supersamflat}{L_{1} = [1,0]^\intercal} &=  \cPr{\supersamflat}{L_{1} = [0,1]^\intercal}\\
&= \frac{1}{2} \kappa(\supersamflat) \\
&=\frac{1}{2}\EE_{J \dist \unif{[2n]_{n+1}}} \frac{1}{n+1} \sum_{j\in J}  P_e(\supersamflat_{J-\{j\}};\supersamflat_j).
\end{aligned}
\end{equation}
Then, by the fact that the $L_1$ can take values in $\{[1,0]^\intercal,[0,0]^\intercal,[0,1]^\intercal\}$ and \cref{eq:sym-prob-entropy} we obtain
\[
\entr{L_1\vert\supersamflat} &= \EE [ -(1-\kappa(\supersamflat)) \log(1- \kappa(\supersamflat))  -  {\kappa(\supersamflat)}/{2}\log ( {\kappa(\supersamflat)}/{2}) -  {\kappa(\supersamflat)}/{2}\log ({\kappa(\supersamflat)}/{2})] \nonumber\\
&=\EE [\binaryentr{ \kappa(\supersamflat)} + \kappa(\supersamflat) \log(2)] \label{eq:entr-firstterm},
\]
where $\binaryentr{\cdot}$ denotes the binary entropy function.

Finally, we find a closed form expression for $\entr{L_1 \vert \supersamflat,\pi,U}$. Since it is assumed that $\Alg$ is consistent we have $L_{U_i,i}=0$ \as. Therefore
$\entr{L_1 \vert \supersamflat,\pi,U}=\entr{L_{\bar{U_1},1} \vert \supersamflat,\pi,U}$. Then, we can write 
\[
\entr{L_{\bar{U_1},1} \vert \supersamflat,\pi,U} &= \EE [\centr{\supersamflat,\pi,U}{L_{\bar{U_1},1}}] \nonumber\\
&=\EE \big[ \cEE{\supersamflat}{[\centr{\supersamflat,\pi,U}{L_{\bar{U_1},1}}]} \big] \nonumber\\
&= \EE \big[  \cEE{\supersamflat}{\binaryentr{[P_e((\supersamflat^{\pi}_{U_1,1},\dots,\supersamflat^{\pi}_{U_n,n});\supersamflat^{\pi}_{\bar{U_1},1})}]} \big] \nonumber\\
&=\EE \big[ \EE_{J \dist \unif{[2n]_{n+1}}} \frac{1}{n+1} \sum_{j \in J} \binaryentr{P_e(\supersamflat_{J-\{j\}};\supersamflat_j)} \big], \label{eq:entr-secondterm}
\]
where in the last line we have used \cref{lem:fun-sym}. Finally by combining \cref{eq:entr-firstterm}, \cref{eq:entr-secondterm}, and \cref{eq:chainrule-ecmi} we obtain the stated result in \cref{eq:ecmi-oneinclusion}.

\section{Description of the One-Inclusion Graph Prediction Algorithm}
\label{app:description-oneinclusion}
In this part we provide a short description of the one-inclusion transductive learning algorithm of \citet{haussler1994predicting}.
Let $\bar{X}=(x_1,\dots,x_n) \in \inspace^n$ and let $\HS$ be a concept class with VC dimension $\vcd$. Let $\HS_{\vert \bar{X}}$ be the equivalence class induced by $\HS$ on the instances given in $\bar{X}$. Similarly for a classifier $h \in \HS$, we can define $h_{\vert \bar{X}}$ as the restriction of $h$ to the instances in $\bar{X}$.  For every $h \in \HS_{\vert \bar{X}}$, let $v_h=(h(x_1),\dots,h(x_n))\in \{0,1\}^n$.  \citet{haussler1994predicting} defined the one-inclusion graph of $\bar{X}$ denoted by $\oneinclusiongraph(\bar{X})=(V,E)$ as follows. $\oneinclusiongraph(\bar{X})$ has the vertex set $ V= \{v_h: h\in \HS_{\vert \bar{X}}\}$, and $(v_h,v_{h'}) \in E$ if and only if the hamming distance of $v_h$ and $v_{h'}$ is one. For an example of $\oneinclusiongraph(\bar{X})$ for the concept class of intervals in one dimension, see Fig.~1 of \citep{haussler1994predicting}. Consider \defn{probability assignment mapping} $\probassignment: E \times V \to [0,1]$ such that for each edge $e$ incident to $v_{h}$ and $v_{h'}$ the following two conditions holds. 
\begin{enumerate}
    \item for all $h'' \in \HS_{\vert \bar{X}}$ with $h'' \neq h$ and $h'' \neq h'$, we have $\probassignment(e,v_{h''})=0$.
    \item  $\probassignment(e,v_h)\geq 0$, $f(e,v_{h'})\geq 0$, and $\probassignment(e,v_h)+\probassignment(e,v_{h'})=1$.
\end{enumerate}
For every $i \in [n]$ and $h \in \HS_{  \vert \bar{X}}$, let
$c_{i,h} \subset \HS_{  \vert \bar{X}}$  be the set of all the hypotheses in  $\HS_{\vert \bar{X}}$ whose restriction to $\bar{X} \setminus \{x_i\}$ equals to $h_{ \vert \bar{X}\setminus\{x_i\}}$. Let $\targetfun \in \HS$, consider a realizable $S=((x_1,y_1),\dots,(x_n,y_n))$ where $y_i=\targetfun(x_i)$ for $i \in \range{n}$. Assume  $S_{n-1}=((x_1,y_1),\dots,(x_{n-1},y_{n-1}))$ is given to a learner and the learner aims to predict the label of $x_n$.  It is immediate to see that the set of hypotheses consistent with $S_{n-1}$ is $c_{n,h^\star}$. Clearly, $|c_{n,h^\star}|\in \{1,2\}$ and if $|c_{n,h^\star}|=1$, we know that
the target hypothesis is $h^\star$. But, what should the learner do when $|c_{n,h^\star}|=2$?

Using $\oneinclusiongraph(\bar{X})$ we can think of the case $|c_{n,h^\star}|=2$ as there is a vertex $v_{h'}$  adjacent to $v_{\targetfun}$, and $v_{h'}$ and $v_{\targetfun}$ differ in $n$-th position. Assume $|c_{n,h^\star}|=2$ and $e_{n,h^\star} = \{\targetfun,h'\}$. Using the probability assignment mapping $\probassignment$ of $\oneinclusiongraph(\bar{X})$, the strategy proposed by \citet{haussler1994predicting} predicts the label $x_n$ to be $h'(x_n)$ with probability $\probassignment(e,v_{\targetfun})$, and to be $\targetfun(x_n)$ with probability $\probassignment(e,v_{h'})$.  By identifying a deep  combinatorial property of $\oneinclusiongraph(\bar{X})$,  Thm.~2.3 in \citep{haussler1994predicting} shows that there exists a probability assignment mapping $\probassignment$ with the mentioned properties such that $\sum_{e \in E} \probassignment(e,v_{h})\leq \vcd$ for all $h \in \oneinclusiongraph(\bar{X})$, and finding such a mapping is computationally easy. Moreover, 
\citet{haussler1994predicting} show there exists \emph{deterministic} probability assignment for every one-inclusion graph such that $\probassignment(e,v_{h})\in \{0,1\}$ for all $e \in E$ and $v_h \in V$, and $\sum_{e \in E} \probassignment(e,v_{h})\leq \vcd$ for all $h \in \oneinclusiongraph(\bar{X})$.

\end{document}